\documentclass[a4paper,11pt]{article} 
\RequirePackage[top=0.75in, bottom=0.75in, left=0.75in, right=0.75in]{geometry}
\usepackage{multirow}
\usepackage{amsmath,amssymb,amsthm}
\usepackage{bm}
\usepackage[colorlinks=true,bookmarks=false,citecolor=blue,urlcolor=blue]{hyperref} 
\usepackage{natbib}
\usepackage{tikz}
\usetikzlibrary{arrows.meta,positioning,calc,shapes.geometric}
\usepackage{enumitem}
\RequirePackage{graphicx}
\usepackage{subcaption}
\usepackage{booktabs}

\newtheorem{theorem}{Theorem}[section]

\newtheorem{proposition}[theorem]{Proposition}
\newtheorem{lemma}[theorem]{Lemma}
\newtheorem{corollary}[theorem]{Corollary}

\newtheorem{definition}[theorem]{Definition}

\newtheorem{remark}[theorem]{Remark}

\newcommand{\JEL}[1]{\par\noindent\textbf{JEL Classification:} #1}

\allowdisplaybreaks

\title{Pareto and Bowley Reinsurance Games in Peer-to-Peer Insurance}

\author{Tim J. Boonen\thanks{Department of Statistics and Actuarial Science, School of Computing and Data Science, The University of Hong Kong, Hong Kong; Email: \texttt{tjboonen@hku.hk}}\hspace{1em}Kenneth Tsz Hin Ng\thanks{Department of Mathematics, The Ohio State University, Columbus,  United States; Email: \texttt{ng.499@osu.edu} }\hspace{1em} Tak Wa Ng\thanks{\'{E}cole d'Actuariat, Université Laval, Quebec, Canada; Email: \texttt{tak-wa.ng.1@ulaval.ca}}\hspace{1em} Thai Nguyen\thanks{\'{E}cole d'Actuariat, Université Laval, Quebec, Canada; Email:  \texttt{thai.nguyen@act.ulaval.ca}}}
\date{\today}

\begin{document}

\maketitle

\begin{abstract}   



We propose a peer-to-peer (P2P) insurance scheme comprising a risk-sharing pool and a reinsurer. A plan manager determines how risks are allocated among members and ceded to the reinsurer, {while the reinsurer sets the reinsurance loading}. {Our work focuses on the strategic interaction between the plan manager and the reinsurer}, and this focus leads to two game-theoretic contract designs: a Pareto design and a Bowley design, for which we derive closed-form optimal contracts. In the Pareto design, cooperation between the reinsurer and the plan manager leads to multiple Pareto-optimal contracts, which are further refined by introducing the notion of coalitional stability. In contrast, the Bowley design yields a unique optimal contract through a leader–follower framework, and we provide a rigorous verification of the individual rationality constraints via pointwise comparisons of payoff vectors. 
Comparing the two designs, we prove that the Bowley-optimal contract is never Pareto optimal and typically yields lower total welfare.
{In our numerical examples, the presence of reinsurance improves welfare, especially with Pareto designs and a less risk-averse reinsurer. 
We further analyze the impact of the single-loading restriction, which disproportionately favors members with riskier losses. } 

\end{abstract}

\noindent\textbf{Keywords:} 
Risk management, 
Peer-to-peer risk sharing, 
Optimal reinsurance, 
 Pareto optimality,
 Cooperative game theory.

\JEL{C71, C72, G22}

\section{Introduction}

Decentralized insurance arrangements, such as peer-to-peer (P2P) insurance \citep{denuit2020investing,feng2023decentralized}, mutual insurance \citep{laux2010financing,li2025mean}, and tontines \citep{milevsky2015optimal,ng2024individual}, have attracted increasing attention as alternatives to traditional centralized insurance models \citep[e.g.,][]{boonen2024pareto}. By allowing members to share risks within a community directly, these schemes promise improved transparency, reduced administrative costs, and better alignment of incentives. In particular, many P2P insurance schemes incorporate an external reinsurance layer to enhance risk coverage and ensure solvency, leading to hybrid designs that combine internal risk sharing with external risk transfer. 
For example, \cite{denuit2021stop} 
study
the stop-loss treaty in P2P insurance in which the retained part is shared by the conditional mean risk-sharing rule \citep{denuit2012convex}. \cite{anthropelos2026expansion} consider the expansion of risk-sharing pools and examine the impact of exogenous reinsurance.
However, to our knowledge, the literature on the strategic interaction between a reinsurer and a decentralized P2P pool remains largely unexplored.

Understanding this interaction is important because, in practice, reinsurers are not passive risk absorbers but strategic market participants who price coverage and influence the feasibility of P2P schemes. Whether the reinsurer acts as a price-setting leader or engages in coordinated contract design directly affects equilibrium risk allocation, premium levels, and the long-run sustainability of decentralized insurance pools. Accordingly, the present paper focuses on this strategic interaction and its implications for all stakeholders in the P2P insurance scheme. 

\subsection{Contributions}


This paper considers a novel mean-variance P2P insurance-reinsurance scheme comprising a risk-sharing pool managed by a plan manager and a reinsurer that provides external insurance protection to the pool.
In this scheme, the P2P plan manager decides the internal risk mutualization among members and external risk transfer to the reinsurer, while the reinsurer sets the price of the reinsurance contract under the expected value premium principle. Depending on the interaction between the reinsurer and the plan manager, we consider two  institutional designs: the Pareto and the Bowley models. 

In the Pareto design, the reinsurer and the P2P plan manager optimize their joint interests, which we term a  \textit{joint-Pareto-optimal} (JP-optimal) contract. Following the two-step approach described in \citet[Section 2]{asimit2021risk}, we determine the risk mutualization between members and the reinsurer via a 
social-objective optimization. Since the price of the reinsurance contract cancels out when aggregating the members' and the reinsurer's disutility, the optimal risk-mutualization and reinsurance strategy is independent of the price of the reinsurance contract,
which gives rise to multiple JP-optimal contracts due to the flexibility to choose the reinsurance premium. This freedom permits additional selection criteria within the JP-optimal set based on coalitional stability considerations.
Specifically, we study the premium decision via a coalition game that assigns to each agent\footnote{In this paper, we refer to ``members'' as participants in the risk-sharing pool, and to ``agents'' as all individuals, including both members and the reinsurer.} the welfare gain from joining the P2P contract. We prove that the \textit{core} is non-empty and that its elements correspond to coalitionally stable and JP-optimal contracts. Furthermore, we derive sufficient conditions for contracts with nonnegative safety loadings, and in particular for contracts with a single common safety loading.


The Bowley design is formulated as a leader–follower game, 
in which the reinsurer acts as the leader and sets the reinsurance premium, while the plan manager responds by determining the reinsurance and internal risk-sharing strategies for the members. By backward induction and solving the two subproblems sequentially, we derive the unique Bowley-optimal contract in closed form, both with and without the single-loading restriction. Unlike the Pareto design, the reinsurance premium in the Bowley framework is uniquely determined,
but the individual rationality (IR) constraints are not immediately satisfied. We therefore provide explicit closed-form conditions on the market parameters that ensure the IR constraints are satisfied.
Finally, by comparing the two game-theoretic designs, we show that the Bowley contract is never JP-optimal and always results in lower total welfare for all agents.

{Through a comprehensive numerical analysis, we compare 
various contracts
and investigate their impacts on welfare. Owing to the multiplicity of admissible safety loadings under the Pareto design, we select JP-optimal contracts that}
Pareto-dominate the Bowley counterparts. 
The numerical comparison of all contracts yields four key insights.
First, Pareto and Bowley contracts dominate the no-reinsurance case in terms of each agent's welfare improvement. This underscores the value of the reinsurance layer, aligning with the finding in \cite{anthropelos2026expansion} despite different settings and research questions.

Second, the Pareto design generally leads to
lower safety loadings and higher risk transfers for all members, and the combined effect leads to a higher premium payment. 
From the reinsurer’s perspective, the additional risk borne under the Pareto contracts is well-compensated by the increased premium income,
resulting in higher welfare for all agents in the community.  

{Third,
we find that the single-loading restriction primarily benefits members with riskier losses, as they are typically charged lower safety loadings than in the unrestricted case. As a result, 
the high-risk member may receive a disproportionately large welfare gain 
at the expense of other members and the reinsurer. Nonetheless, this does not necessarily reduce total welfare relative to contracts without the restriction: with a less risk-averse reinsurer, Bowley contracts with the single-loading restriction achieve a larger total welfare improvement.

Finally, through
comparative statics analysis of the reinsurer’s risk aversion, we find that {without the single-loading restriction}, Bowley-optimal contracts exhibit underinsurance compared to JP-optimal cases, echoing the results in \cite{jiang2025bowley}. In addition, the welfare improvement depends on the reinsurer's risk tolerance: when she 
is less risk averse and thus more risk is transferred from the pool, a greater increase in total welfare results.

\subsection{Related Literature}

Our work contributes to the literature on decentralized insurance.
This strand of research has followed several representative directions, such as the construction of risk-sharing architectures through, among others, axiomatic characterizations \citep{denuit2022risk} and optimization-based approaches \citep{abdikerimova2022peer}. Other contributions from insurance economics examine issues such as adverse selection \citep{chen2024cost} and moral hazard \citep{BoonenNgNguyen2025}, as well as institutional design aspects. For example, \cite{denuit2021risk} propose three business models with different governance, and \cite{clemente2023optimal} design a novel cashback mechanism based on the Shapley value. In particular, our paper devises a new P2P insurance scheme featuring a proportional risk mutualization similar to that in \citet{feng2023peer} and an additional proportional reinsurance treaty.
Unlike \cite{denuit2021stop} and \cite{anthropelos2026expansion}, who assume that the reinsurance layer is exogenously offered, the novelty of this paper lies in incorporating the reinsurer’s perspective and modeling her strategic interaction with the P2P insurance manager in setting the reinsurance premium, which fundamentally shapes equilibrium risk sharing
and welfare outcomes.

Besides, our Pareto and Bowley reinsurance games connect the literature on decentralized insurance to game-theoretic (re)insurance contracting.\footnote{The literature review focuses on (re)insurance contracting under Pareto and Bowley games. For a more comprehensive review of reinsurance contracting, we refer to \cite{cai2020optimal}.} In Pareto-optimal (re)insurance arrangements, there are mainly two categories related to premium setting. In the first category, the contract is composed via a weighted-sum optimization with a given premium function. For instance, \cite{cai2017pareto} develop a mutually acceptable Pareto-optimal reinsurance contract that accommodates both the reinsurer and the insurer's goals. \cite{lo2019pareto} further extend their work to incorporate risk constraints.

Another category endogenizes the premium decision in the bargaining process among agents, which is in line with the seminal contribution of \cite{raviv1979design}. 
Methodologically, this approach first optimizes the total welfare to determine the optimal indemnity, and then redistributes the resulting welfare gain among agents via a cooperative game to set the premium. Our Pareto design features a cooperative game study on premium decision, further advancing this line of research in the P2P insurance setting.
To name a few contributions in this spirit, \cite{asimit2018insurance} study Pareto-efficient insurance contracts with one policyholder and multiple insurers, which is further examined in \cite{boonen2025pareto} with distributional uncertainty of the insurable loss. \cite{boonen2024pareto} explore the risk sharing between one monopolistic insurer and multiple policyholders with applications in flood risk management. Our Pareto design adds a risk-mutualization layer to the setting of \cite{boonen2024pareto}. Therein, the core of the insurance game exists trivially since the insurer serves as the veto player who establishes the risk-sharing scheme. In contrast, in our model, the risk-sharing scheme can operate even without the reinsurer's involvement. Therefore, our work contributes by
{providing a nontrivial} proof of the non-emptiness of the resulting core \citep{gillies1953some}.

The leader–follower framework is another key model in (re)insurance contracting. 
In the continuous-time setting, \cite{chen2018new} study such a reinsurance game under the expected utility maximization criterion, which is further extended to the settings of ambiguous claim arrival \citep{hu2018robust}, multiple competitive followers with relative performance concern \citep{bai2022hybrid}, and multiple leaders with different premium principles \citep{cao2023reinsurance}.

The Bowley reinsurance game considered herein belongs to the same leader-follower framework, featuring a monopolistic reinsurer in a discrete-time setting. The seminal work of this stream dates back to the study in  \cite{chan1985reinsurer} under the expected utility setting. \cite{cheung2019risk} revisit the problem under a general premium principle for the reinsurer and use a distortion risk measure for the insurer. In addition, \cite{boonen2023bowley} compare Pareto and Bowley-optimal reinsurance designs. They find that the Bowley optimum can be Pareto efficient but leaves the insurer's welfare indifferent to the status quo. However, the Bowley optimum is not necessarily Pareto optimal: 
\cite{jiang2025bowley} also make such a comparison under the generalized mean-variance preference setting and prove the inefficiency
of Bowley-optimal contracts. Our result echoes this finding in the literature on P2P insurance design and suggests caution regarding the adoption of the Bowley reinsurance scheme.


The remainder of the paper is organized as follows. Section \ref{Sec:ModelFormulations} formulates the general model setting. Pareto and Bowley designs are introduced in Section \ref{Sec:ParetoDesign} and Section \ref{Sec:BowleyDesign}, respectively. Section \ref{Sec:Numeric} numerically illustrates two designs and performs a welfare analysis. We conclude the paper in Section \ref{Sec:Conclusion} with future research directions. Most proofs and additional exposition are relegated to the Appendix.

\section{Model Formulations}\label{Sec:ModelFormulations}


\subsection{Basic setting and risk sharing}

We consider a pool of $n\geq2$ members who share individual-specific losses $\bm{X}=(X_1,\dots,X_n)$, where $X_i$ represents the loss random variable of the Member $i$. Denote the mean and the positive definite covariance matrix of $\bm{X}$ by $\bm{\mu}=(\mu_1,\dots,\mu_n)$ and $\bm{\Sigma}=(\sigma_{ij})_{i,j=1,\dots,n}$, respectively, where ${\mu_i},\sigma_{ii}>0$ for any $i=1,\dots,n$.
In the sequel, unless otherwise specified, we use bold symbols to denote vectors or matrices; and unbolded symbols to denote scalars. 

A plan manager oversees the risk management strategies of the pool, which consist of two main components:
    \begin{enumerate}
        \item \textbf{Reinsurance}. The plan manager decides the proportional reinsurance strategy $\bm{p}=(p_1,\dots,p_n)$, which transfers the risk $\bm{X}$ partially to a reinsurer.\footnote{We use the term ``reinsurer" to distinguish it from the P2P plan manager. Such contracts can also be offered by insurers. } In return, each member $i$ pays a premium $\pi_i$ for the reinsurance policy. We assume that the expected value premium principle is adopted:
            \begin{equation*}
    \pi_i(\eta_i,p_i):=(1+\eta_i)\mathbb{E}(p_iX_i)=(1+\eta_i)p_i\mu_i,
\end{equation*}
    where $\eta_i\geq 0$ is the safety loading.  
        \item \textbf{Risk mutualization}. The plan manager also determines the  risk mutualization rule $\bm{A}=(a_{ij})_{i,j\in\{1,\dots,n\}}$ on how losses are shared among members of the pool, where $a_{ij}$ represents the portion of member $j$'s loss, $a_{ij}X_j$, which is borne by Member $i$ \textit{ex post}.   
    \end{enumerate}
Summarizing the above, Member $i$ has to pay $\pi_i(\eta_i,p_i)$ {to} cede $p_iX_i$ to the reinsurer, and has to bear $\sum_{j=1}^na_{ij}X_j$ from other members {and herself}. 
The mechanism is graphically illustrated in Fig. \ref{Fig:LSXi}.
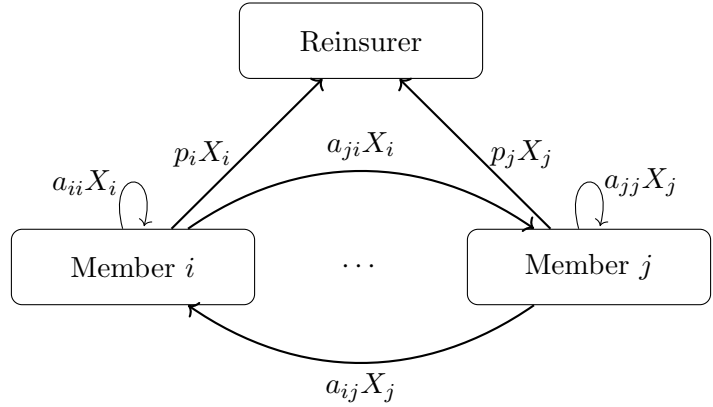
\begin{figure}[htbp]
\centering
\begin{tikzpicture}[
    node distance=6cm,
    mainnode/.style={
        draw,
        rounded corners,
        minimum width=3.2cm,
        minimum height=1cm,
        align=center
    },
    label/.style={
        draw=none,
        fill=none,
        font=\normalsize
    },
    dashedlink/.style={dashed, thick},
    arrow/.style={->, thick},
    loopstyle/.style={looseness=10}
]

\node (mi) [mainnode] {Member $i$};
\node (mj) [right of=mi, mainnode] {Member $j$};
\node (dots) at ($(mi)!0.5!(mj)$) {$\cdots$};
\node[mainnode] (r) at ($(mi)!0.5!(mj)+(0,3cm)$) {Reinsurer};
\draw[arrow] (mi) to[bend left=35]
    node[label, midway, above] {$a_{ji}X_i$} (mj);

\draw[arrow] (mj) to[bend left=35]
    node[label, midway, below] {$a_{ij}X_j$} (mi);

\draw[loopstyle] (mi)
    edge[loop above] node[label, left] {$a_{ii}X_i\,$} (mi);

\draw[loopstyle] (mj)
    edge[loop above] node[label, right] {$\, a_{jj}X_j$} (mj);

\draw[arrow] (mi) -- node[label, midway, left] {$p_iX_i\,$} (r);
\draw[arrow] (mj) -- node[label, midway, right] {$\,p_jX_j$} (r);

\end{tikzpicture}

\caption{Illustration of the risk management mechanism.}\label{Fig:LSXi}
\end{figure}

\vspace{-1em}

The risk-sharing scheme herein adds an additional  reinsurance layer to the P2P risk mutualization introduced in \cite{feng2023peer}. Motivated by their work, we formulate the zero-conserving condition and actuarial fairness for our mechanism as follows. 

\begin{definition}[Zero conserving]
    The risk-mutualization-reinsurance rule $(\bm{A},\bm{p})$ is said to be zero conserving if all losses are shared completely among members and the reinsurer. Mathematically, $p_i+\sum_{j=1}^na_{ji}=1$ for all $i=1,\dots,n$. 
The condition can be equivalently expressed as $\mathbf{1}^\top\bm{A}+\bm{p}^\top=\bm{1}^\top$, where $\bm{1}\in\mathbb{R}^n$ is the vector with all entries being 1. 
\end{definition}


\begin{definition}[Actuarial fairness]
    The risk-mutualization-reinsurance rule $(\bm{A},\bm{p})$ is said to be actuarially fair if the post-reinsurance expected loss is preserved before and after risk mutualization for all members in the pool. Mathematically, $\sum_{j=1}^na_{ij}\mu_j=(1-p_i)\mu_i$ for all $i=1,\dots,n$. 
Equivalently, the condition can be expressed as $\bm{A}\bm{\mu} + \bm{D}(\bm{\mu})\bm{p}=\bm{\mu}$, where $\bm{D}$ is the operator that transforms an $n$-dimensional vector to an $n\times n$ diagonal matrix. 
\end{definition}



On the other hand, the reinsurer (indexed by $R$ in the sequel) shall determine the safety loadings $\bm{\eta}=(\eta_1,\dots,\eta_n)$ charged to the members based on the collective risk exposure. The total premium collected by the reinsurer is thus $\sum_{i=1}^n\pi_i(\eta_i,p_i)=(\bm{D}(\bm{\mu}) (\bm{1}+\bm{\eta}))^\top\bm{p}$. 
 Here, we assume that the reinsurer can charge members differently based on their individual risks. We collect the notions of risk sharing among members, reinsurance, and safety loading to define the P2P insurance contract below.

\begin{definition}[P2P insurance contract]
    A triplet $(\bm{A},\bm{p},\bm{\eta})$ is called a P2P insurance contract, where $\bm{A}$ is the risk mutualization among members, $\bm{p}$ is the proportional reinsurance strategy, and $\bm{\eta}$ is the safety loading factor.
\end{definition}

\subsection{Agents' preferences}
We introduce the preferences of each individual member, the plan manager, and the reinsurer as follows. {The risk borne by each member} is measured by a mean-variance disutility based on the risk mutualization arrangement:
\begin{equation*}
    \rho_i(\bm{A}):=\sum_{j=1}^n\mathbb{E}(a_{ij}X_j)+\frac{\gamma_i}{2}Var\left(\sum_{j=1}^na_{ij}X_j\right)=\bm{A}_i\bm{\mu}+\frac{\gamma_i}{2}\bm{A}_i\bm{\Sigma}\bm{A}_i^\top,\quad\text{for all }i=1,\cdots,n,
    \end{equation*}
    where $\bm{A}_i$ is the $i$-th row of $\bm{A}$
    and {$\gamma_i > 0$} captures the heterogeneous risk aversion of Member $i$. Together with the premium paid to the reinsurer for the reinsurance policy, $\pi_i(\eta_i,p_i)$, the preference for the $i$-th member is given by $u_i(\bm{A},\bm{p},\bm{\eta}):=\rho_i(\bm{A})+\pi_i(\eta_i,p_i)$.
The plan manager's preference is then given by the sum of the members' preferences:
\begin{equation*}
    u(\bm{A},\bm{p},\bm{\eta}):=\sum_{i=1}^nu_i(\bm{A},\bm{p},\bm{\eta})=\bm{1}^\top\bm{A}\bm{\mu}+\frac{1}{2}tr(\bm{D}(\bm{\gamma})\bm{A}\bm{\Sigma}\bm{A}^\top)+(\bm{D}(\bm{\mu})(\bm{1}+\bm{\eta}) )^\top\bm{p},
\end{equation*}
where $tr(\cdot)$ is the trace operator.

The reinsurer aims to maximize the profit while controlling the risk taken. The former is given by $\sum_{i=1}^n\pi_i(\eta_i,p_i)$, while the latter is measured by the mean-variance disutility:
\begin{equation*}
    \rho_R(\bm{p}):=\mathbb{E}\left(\sum_{i=1}^np_iX_i\right)+\frac{\gamma_{R}}{2}Var\left(\sum_{i=1}^np_iX_i\right)=\bm{\mu}^\top\bm{p}+\frac{\gamma_{R}}{2}\bm{p}^\top\bm{\Sigma}\bm{p},
\end{equation*}
where $\gamma_{R}\geq 0$ captures the reinsurer's risk aversion, with a larger $\gamma_{R}$ signifying that the reinsurer is more risk averse. The reinsurer's preference is thus given by the sum of the two components: 
\begin{equation}\label{Eq:ReinsurerPreference}
    v(\bm{\eta},\bm{p}):=\rho_R(\bm{p})-\sum_{i=1}^n\pi_i(\eta_i,p_i)
    =\frac{\gamma_{R}}{2}\bm{p}^\top\bm{\Sigma}\bm{p}-(\bm{D}(\bm{\mu})\bm{\eta})^\top\bm{p}.
\end{equation}

\subsection{Individual rationality}\label{Sec:IR}

The individual rationality (IR) constraint is a set of conditions that must be satisfied so that all agents (members and the reinsurer) are better off than under the status quo, thereby incentivizing them to join the risk-sharing scheme. For members in the pool, we define the welfare gain $\omega_i$ for Member $i$ by
\begin{equation}
\label{Eq:wi}
    \omega_i(\bm{A},\bm{p},\bm{\eta}) := \mu_i + \frac{\gamma_i}{2}\sigma_i^2 - u_i(\bm{A},\bm{p},\bm{\eta})=\frac{\gamma_i}{2}(\sigma_i^2-\bm{A}_i\bm{\Sigma}\bm{A}_i^\top)-\eta_ip_i\mu_i, \quad i=1,\dots,n, 
\end{equation}
where the second equality follows from the actuarial fairness condition.
This represents the reduction in the member's disutility after joining the pool. The IR constraints for the members are then given by 
    \begin{equation}
        \label{Eq:MemberIR}
        \omega_i(\bm{A},\bm{p},\bm{\eta}) \geq 0\iff\eta_ip_i\mu_i\leq\frac{\gamma_i}{2}(\sigma_i^2-\bm{A}_i\bm{\Sigma}\bm{A}_i^\top), \quad \text{for all }i=1,\dots,n. 
    \end{equation}

The reinsurer's welfare gain can be formulated similarly as 
    \begin{equation*}
        \omega_R(\bm{A},\bm{p},\bm{\eta})  := -v(\bm{p},\bm{\eta}) = (\bm{D}(\bm{\mu})\bm{\eta} )^\top\bm{p} - \frac{\gamma_{R}}{2}\bm{p}^\top\bm{\Sigma}\bm{p},
    \end{equation*}
which again measures the reduction in disutility upon forming the risk-sharing scheme. The IR constraint for the reinsurer is thus 
\begin{equation}\label{Eq:ReinsurerIR}
\omega_R(\bm{A},\bm{p},\bm{\eta})  \geq 0  \iff
\frac{\gamma_{R}}{2}\bm{p}^\top\bm{\Sigma}\bm{p}\leq(\bm{D}(\bm{\mu})\bm{\eta} )^\top\bm{p}.
\end{equation}
In the sequel, we denote by $\mathcal{IR}$ the set of all contracts $(\bm{A},\bm{p},\bm{\eta})$ that satisfy constraints \eqref{Eq:MemberIR} and \eqref{Eq:ReinsurerIR}. 




\section{Pareto Design}\label{Sec:ParetoDesign}


In this section, we consider the situation in which the reinsurer and the plan manager cooperate on the risk-sharing scheme, thereby leading to a Pareto game. We use the term Joint-Pareto (JP) optimality to distinguish it from Pareto optimality in Definition \ref{Def:PO} {below}.

\begin{definition}\label{Def:JPO}[Joint-Pareto-optimal P2P insurance contract]
    A contract $(\bm{A},\bm{p},\bm{\eta})$ is called Joint-Pareto optimal if {$(\bm{A},\bm{p},\bm{\eta})\in\mathcal{IR}$, and} no   $(\bm{\Tilde{A}},\bm{\Tilde{p}},\bm{\Tilde{\eta}}){\in\mathcal{IR}}$ satisfies  $v(\bm{\Tilde{\eta}},\bm{\Tilde{p}})\leq v(\bm{\eta},\bm{p})$ and $u_i(\bm{\Tilde{A}},\bm{\Tilde{p}},\bm{\Tilde{\eta}})\leq u_i(\bm{A},\bm{p},\bm{\eta})$
for all $i=1,\dots,n$, with at least one of these inequalities being strict.
\end{definition}

\subsection{Characterization of Joint-Pareto optimality}\label{Sec:CJPO}

To characterize the JP-optimal contracts in Definition \ref{Def:JPO}, we adopt the 2-step approach delineated in \citet[Section 2]{asimit2021risk}. First, we solve 
\begin{equation}\label{Prob:RS}
    \min_{\bm{A},\bm{p}}\left\{\rho_R(\bm{p})+\sum_{i=1}^n\rho_i(\bm{A})\right\}\quad
    \text{s.t. }\bm{A}\bm{\mu}+\bm{D}(\bm{\mu})\bm{p}=\bm{\mu},\quad \mathbf{1}^\top\bm{A}+\bm{p}^\top=\bm{1}^\top.
\end{equation}
Note that the objective of Problem \eqref{Prob:RS} is given by the sum of members' and the reinsurer's preferences, which is independent of $\bm{\eta}$ as the relevant terms are canceled in the summation. Based on the solution to Problem \eqref{Prob:RS}, we then determine the loading $\bm{\eta}$ in a way that the resulting triplet is JP-optimal.  

We now provide the solution to Problem \eqref{Prob:RS}. {To this end, we define the matrix 
 \begin{equation*}
    \overline{\bm{M}}:=\left(\gamma_{R}+\frac{1}{\sum_{j=1}^n\gamma_j^{-1}}\right)\bm{\Sigma}+k\bm{D}(\bm{\mu}^2)\bm{D}(\bm{\gamma})-\frac{k\bm{\mu}\bm{\mu}^\top}{\sum_{j=1}^n\gamma_j^{-1}},
\end{equation*} 
where $\bm{\mu}^2$ is the componentwise square  $\bm{\mu}$, and $k:=(\bm{\mu}^\top\bm{\Sigma}^{-1}\bm{\mu})^{-1}$, which is positive since $\bm{\mu}\neq \bm{0}$ and $\bm{\Sigma}$ is positive definite. The following lemma addresses the invertibility of $\overline{\bm{M}}$.
\begin{lemma}\label{Lem:Mbar}  The matrix $\overline{\bm{M}}$ is positive definite and thus invertible.
\end{lemma}}

\begin{proof}
    See Appendix \ref{App:Mbar}.
\end{proof}

The following proposition solves Problem \eqref{Prob:RS}.

\begin{proposition}
\label{Prop:ConvexPareto}
  The unique minimizer of Problem \eqref{Prob:RS} is
    \begin{align}\label{Eq:OptAP2}
    \begin{split}
        \bm{A}_*=&\ \frac{\bm{D}(\bm{\gamma})^{-1}\bm{1}\bm{1}^\top}{\sum_{j=1}^n\gamma_j^{-1}}\bm{D}(\bm{1}-\bm{p}_*)+k\left(\bm{I}_n-\frac{\bm{D}(\bm{\gamma})^{-1}\bm{1}\bm{1}^\top}{\sum_{j=1}^n\gamma_j^{-1}}\right)\bm{D}(\bm{1}-\bm{p}_*)\bm{\mu}\bm{\mu}^\top\bm{\Sigma}^{-1},\\
    \bm{p}_*=&\ \bm{1}-\gamma_{R}\overline{\bm{M}}^{-1}\bm{\Sigma}\bm{1}.
    \end{split}
    \end{align}
    In addition, if     \begin{equation}\label{Eq:unicond2}
    \begin{aligned}
        &-\gamma_R \sigma_i^2 +  \sum_{m\neq i}\left(\frac{k\mu_i\mu_m}{\sum_{j=1}^n\gamma_j^{-1}}-\left(\gamma_R+\frac{1}{\sum_{j=1}^n\gamma_j^{-1}}\right)\sigma_{im} \right)_+ \\
        &< \frac{\sum_{{m\neq i}}(k\mu_i\mu_m-\sigma_{im})}{\sum_{j=1}^n\gamma_j^{-1}}   < k\mu_i^2\gamma_i {+ \frac{\sigma_i^2-k\mu_i^2}{\sum_{j=1}^n\gamma_j^{-1}}}  - \sum_{m\neq i} \left( \left(\gamma_R + \frac{1}{\sum_{j=1}^n\gamma_j^{-1}} \right)\sigma_{im}-\frac{k\mu_i\mu_m}{\sum_{j=1}^n\gamma_j^{-1}}  \right)_+ 
    \end{aligned}
\end{equation}
    for any $i\in\{1,\dots,n\}$, then $\bm{p}_*\in(0,1)^n$.
\end{proposition}

\begin{proof}
    See Appendix \ref{App:ConvexPareto}.
\end{proof}

Note that the risk allocation in Proposition \ref{Prop:ConvexPareto} differs from the one depicted in \citet[Remark 3]{feng2023peer} due to the extra reinsurance layer.
One can immediately infer from the form of $\bm{A}_*$ that if there is no risk transfer to the reinsurer, both risk mutualizations in the current paper and that in \cite{feng2023peer} coincide. 

From the form of $\bm{p}_*$, one can deduce that $\gamma_R$ should be positive and sufficiently small to circumvent full and zero reinsurance. In particular, $\gamma_R=0$ leads to $\bm{p}_*=\bm{1}$, i.e., it is in the best interest of all agents for the risk-neutral reinsurer to absorb all risks in the Pareto design.
This entails that the reinsurer's risk tolerance should be within an appropriate range to avoid both full reinsurance and no reinsurance. Moreover, \eqref{Eq:unicond2} {is easier to satisfy} when the $\sigma_{im}$ have small magnitudes compared to $\sigma_i^2$, i.e., $\bm{\Sigma}$ exhibits a weak dependence, which enables the diversification among agents.

The following statement characterizes the Joint-Pareto optimality.

\begin{theorem}\label{Thm:JPO}[Characterization of Joint-Pareto optimality]
    Suppose that Condition \eqref{Eq:unicond2} holds.
    Let $\mathcal{JPO}$ be the set of JP-optimal contracts and $\mathcal{Z}:=\{(\bm{A},\bm{p},\bm{\eta})\in\mathcal{IR}: (\bm{A},\bm{p})\text{ solves Problem \eqref{Prob:RS}}\}$.
Then, $\mathcal{JPO}=\mathcal{Z}\neq\emptyset$. Furthermore, {let $(\bm{A}_*,\bm{p}_*)$ be given in \eqref{Eq:OptAP2}, and $\bm{\eta}_*$ be a vector of safety loadings. Then, the following statements are equivalent: } 

\begin{enumerate}[label=(\roman*)]
\item $(\bm{A}_*,\bm{p}_*,\bm{\eta}_*)\in \mathcal{JPO}$.
\item There exists $(c_1,c_2,\dots,c_n)\in\mathbb{R}^n$ such that $(\bm{A}_*,\bm{p}_*,\bm{\eta}_*)$ solves 
\begin{equation}\label{Prob:equJPO}
\begin{aligned}
        & \min_{(\bm{A},\bm{p},\bm{\eta})\in\mathcal{IR},{\bm{p}>0}}v(\bm{\eta},\bm{p})\\
        \text{s.t. }& u_i(\bm{A},\bm{p},\bm{\eta})=c_i,\ \text{for all }i=1,\dots,n,\ \bm{A}\bm{\mu}+\bm{D}(\bm{\mu})\bm{p}=\bm{\mu},\ \mathbf{1}^\top\bm{A}+\bm{p}^\top=\bm{1}^\top.
    \end{aligned}
\end{equation}
\item There exists $c_R\in\mathbb{R}$ such that $(\bm{A}_*,\bm{p}_*,\bm{\eta}_*)$ solves 
\begin{equation}\label{Prob:equJPO2}
    \min_{(\bm{A},\bm{p},\bm{\eta})\in\mathcal{IR},{\bm{p}>0}}u(\bm{A},\bm{p},\bm{\eta})\text{ s.t. }v(\bm{\eta},\bm{p})=c_R,\ \bm{A}\bm{\mu}+\bm{D}(\bm{\mu})\bm{p}=\bm{\mu},\ \mathbf{1}^\top\bm{A}+\bm{p}^\top=\bm{1}^\top.
\end{equation}
\end{enumerate}
\end{theorem}

\begin{proof}
    See Appendix \ref{App:JPO}.
\end{proof}

Theorem \ref{Thm:JPO} shows that one can always select a loading $\bm{\eta}_*$ such that, along with the solution $(\bm{A}_*,\bm{p}_*)$ to Problem \eqref{Prob:RS}, the resulting contract $(\bm{A}_*,\bm{p}_*,\bm{\eta}_*)$ satisfies all agents' IRs \eqref{Eq:MemberIR} and \eqref{Eq:ReinsurerIR}, which is also a JP-optimal contract. In addition, Theorem \ref{Thm:JPO} offers two equivalent optimizations. By varying parameter $c_i$, $i=1,\ldots,n$, (resp. $c_R$) in (ii) (resp. (iii)), the entire Pareto frontier can be traced.

\subsection{Cooperative game study on premium decision}




Theorem \ref{Thm:JPO} shows that there may exist many safety loadings $\bm{\eta}_*$ such that the triple $(\bm{A}_*, \bm{p}_*, \bm{\eta}_*)$ is JP-optimal. This flexibility arises because all terms involving the safety loading cancel out in Problem \eqref{Prob:RS}. As revealed in Theorem \ref{Thm:JPO} and its proof, the selection of safety loadings can be considered as an allocation problem that apportions welfare gains from forming the P2P risk-sharing scheme. To study such a distribution problem and the corresponding premium decisions systematically, we formulate a cooperative transferable utility (TU) game $(\mathcal{N},\mathcal{B})$, where $\mathcal{N}:=\{1,\dots,n,R\}$ is a finite set including all members and the reinsurer and $\mathcal{B}:2^{\mathcal{N}}\mapsto[0,+\infty)$ represents the worth of coalitions by mapping every subgroup to the corresponding maximal coalitional welfare gain, which is defined in detail below. 

For the subgroup of members $\mathcal{S}\subseteq\mathcal{T}:=\{1,\dots,n\}$, we let $\bm{\mu}^\mathcal{S}:= (\mu_i)_{i\in\mathcal{S}}$, $\bm{\Sigma}^\mathcal{S} := (\sigma_{ij})_{i,j\in\mathcal{S}}$, and $\bm{\gamma}^\mathcal{S}:=(\gamma_i)_{i\in\mathcal{S}}$. For any $\bm{A}^\mathcal{S} = (a_{ij})_{i,j \in \mathcal{S}} \in \mathbb{R}^{|\mathcal{S}|\times|\mathcal{S}|}$, $\bm{p}^\mathcal{S} = (p_i)_{i\in\mathcal{S}} \in \mathbb{R}^{|\mathcal{S}|}$, $\bm{\eta}^\mathcal{S} = (\eta_i)_{i\in\mathcal{S}}\in \mathbb{R}^{|\mathcal{S}|}$, we (re)define, for any $i\in\mathcal{S}$,
    \begin{align*}     &\rho_i(\bm{A}^\mathcal{S};\mathcal{S}):= \sum_{j\in\mathcal{S}} a_{ij}\mu_j + \frac{\gamma_i}{2}Var\left(\sum_{j\in\mathcal{S}} a_{ij}X_j  \right),\ u_i(\bm{A}^\mathcal{S},\bm{p}^\mathcal{S},\bm{\eta}^\mathcal{S};\mathcal{S}\cup\mathcal{R}) := \rho_i(\bm{A}^\mathcal{S};\mathcal{S}) + \pi_i(\eta_i,p_i), \\
     &u(\bm{A}^\mathcal{S},\bm{p}^\mathcal{S},\bm{\eta}^\mathcal{S};\mathcal{S}\cup\mathcal{R}) := \sum_{i\in\mathcal{S}}u_i(\bm{A}^\mathcal{S},\bm{p}^\mathcal{S},\bm{\eta}^\mathcal{S};\mathcal{S}\cup\mathcal{R}), \
     \rho_R(\bm{p}^\mathcal{S};\mathcal{S}\cup\mathcal{R}) := \sum_{i\in\mathcal{S}}p_i\mu_i + \frac{\gamma_R}{2}Var\left(\sum_{i\in\mathcal{S}} p_i X_i \right), \\ 
 & v(\bm{\eta}^\mathcal{S},\bm{p}^\mathcal{S};\mathcal{S}\cup\mathcal{R}) :=  \rho_R(\bm{p}^\mathcal{S};\mathcal{S}\cup\mathcal{R}) - \sum_{i\in\mathcal{S}}\pi_i(\eta_i,p_i). 
    \end{align*}

Also, we define the set of all contracts that meet the individual rationality constraints, actuarial fairness, and zero-conserving condition for the subgroup $\mathcal{S}\cup \mathcal{R}$, {where $\mathcal{R}:=\{R\}$, by}: 
    \begin{equation*}
        \begin{aligned}
    {\mathcal{F}^{\mathcal{S}\cup\mathcal{R}}} := \bigg\{(\bm{A}^{\mathcal{S}\cup\mathcal{R}}&,\bm{p}^{\mathcal{S}\cup\mathcal{R}},\bm{\eta}^{\mathcal{S}\cup\mathcal{R}}) : \omega_i(\bm{A}^{\mathcal{S}\cup\mathcal{R}},\bm{p}^{\mathcal{S}\cup\mathcal{R}},\bm{\eta}^{\mathcal{S}\cup\mathcal{R}};\mathcal{S}\cup\mathcal{R}) \geq 0,   \ i \in\mathcal{S}\cup\mathcal{R}, \\ &\bm{A}^{\mathcal{S}\cup\mathcal{R}}\bm{\mu}^\mathcal{S}+\bm{D}(\bm{\mu}^\mathcal{S})\bm{p}^{\mathcal{S}\cup\mathcal{R}} =\bm{\mu}^\mathcal{S},\ (\bm{1}^\mathcal{S})^\top\bm{A}^{\mathcal{S}\cup\mathcal{R}} + (\bm{p}^{\mathcal{S}\cup\mathcal{R}})^\top =(\bm{1}^\mathcal{S})^\top
 \bigg\}. 
 \end{aligned}
    \end{equation*} 
Note that with a slight abuse of notation, here we assume that $(\bm{A}^{\mathcal{S}\cup\mathcal{R}},\bm{p}^{\mathcal{S}\cup\mathcal{R}},\bm{\eta}^{\mathcal{S}\cup\mathcal{R}})\in \mathbb{R}^{|\mathcal{S}|\times|\mathcal{S}|}\times \mathbb{R}^{|\mathcal{S}|}\times \mathbb{R}^{|\mathcal{S}|}$ for $\mathcal{S}\subseteq \mathcal{T}$. For the coalition $\mathcal{S}\subseteq \mathcal{T}$ without the reinsurer, $\mathcal{F}^\mathcal{S}$ can be defined by forcing $\bm{p}^{\mathcal{S}}= \bm{\eta}^\mathcal{S} =\bm{0}^\mathcal{S}$, where  $\bm{0}^\mathcal{S}$ denotes the zero vector in $\mathbb{R}^{|\mathcal{S}|}$,
i.e.,
\begin{equation*}
\mathcal{F}^{\mathcal{S}} := \{(\bm{A}^\mathcal{S},\bm{0}^\mathcal{S},\bm{0}^\mathcal{S}): \omega_i(\bm{A}^\mathcal{S},{\bm{0}^\mathcal{S},\bm{0}^\mathcal{S}};\mathcal{S}) \geq 0,  \ i \in\mathcal{S},\ \bm{A}^\mathcal{S}\bm{\mu}^\mathcal{S}=\bm{\mu}^\mathcal{S},\  (\bm{1}^\mathcal{S})^\top\bm{A}^\mathcal{S} =(\bm{1}^\mathcal{S})^\top\}.
\end{equation*}

Next, we define the mapping $\mathcal{B}$ for coalitions with and without the reinsurer as follows. For coalitions $\mathcal{C}$ {with the reinsurer included, i.e.,}
 $\mathcal{C}=\mathcal{S}\cup\mathcal{R}$, 
$\mathcal{S}\subseteq \mathcal{T}$, we define 
\begin{equation*}
\mathcal{B}(\mathcal{C}):=\mathcal{B}(\mathcal{S}\cup\mathcal{R})=\sum_{i\in\mathcal{S}}\left(\mu_i+\frac{\gamma_i\sigma_i^2}{2}\right)- \sum_{i\in\mathcal{S}}\rho_i(\bm{A}^{\mathcal{S}\cup\mathcal{R}}_*;\mathcal{S})-\rho_R(\bm{p}^{\mathcal{S}\cup\mathcal{R}}_*;\mathcal{S}\cup\mathcal{R}),
\end{equation*}
where $(\bm{A}^{\mathcal{S}\cup\mathcal{R}}_*,\bm{p}^{\mathcal{S}\cup\mathcal{R}}_*)\in \mathbb{R}^{|\mathcal{S}|\times|\mathcal{S}|}\times \mathbb{R}^{|\mathcal{S}|}$ is the  solution to the following minimization problem, which is the version of Problem \eqref{Prob:RS} under the coalition $\mathcal{C}=\mathcal{S}\cup\mathcal{R}$, 
\begin{equation}\label{Prob:RS:C}
    \begin{aligned}
& \min_{\bm{A}^{\mathcal{S}\cup\mathcal{R}},\bm{p}^{\mathcal{S}\cup\mathcal{R}}}\left\{\rho_R(\bm{p}^{\mathcal{S}\cup\mathcal{R}};\mathcal{S}\cup\mathcal{R})+\sum_{i\in \mathcal{S}  }\rho_i(\bm{A}^{\mathcal{S}\cup\mathcal{R}};\mathcal{S})\right\}\\ 
    & \text{s.t. }\bm{A}^{\mathcal{S}\cup\mathcal{R}}\bm{\mu}^{\mathcal{S}}+\bm{D}(\bm{\mu}^{\mathcal{S}})\bm{p}^{\mathcal{S}\cup\mathcal{R}}=\bm{\mu}^{\mathcal{S}},\quad (\mathbf{1}^{\mathcal{S}})^\top\bm{A}^{\mathcal{S}\cup\mathcal{R}}+(\bm{p}^{\mathcal{S}\cup\mathcal{R}})^\top=(\bm{1}^{\mathcal{S}})^\top.
    \end{aligned}
\end{equation}
In particular, 
    \begin{equation*}
        \mathcal{B}(\mathcal{N}) = \sum_{i=1}^n \left(\mu_i + \frac{\gamma_i\sigma_i^2}{2}\right) - \sum_{i=1}^n\rho_i(\bm{A}_*;\mathcal{T}) - \rho_R(\bm{p}_*;\mathcal{N}),
    \end{equation*}
where $(\bm{A}_*,\bm{p}_*)$ is the solution of Problem \eqref{Prob:RS}. 

On the other hand, {for coalitions $\mathcal{C}$ without the reinsurer, i.e., $R\notin\mathcal{C}$,}
we define $\mathcal{B}(\mathcal{C})$ by considering the following version of Problem \eqref{Prob:RS} under $\mathcal{C}$, i.e., 
\begin{equation}\label{Prob:RS2}
    \min_{\bm{A}^\mathcal{C}} \sum_{i\in\mathcal{C}}\rho_i(\bm{A}^\mathcal{C};\mathcal{C}) \quad 
    \text{s.t. }\bm{A}^\mathcal{C}\bm{\mu}^\mathcal{C} =\bm{\mu}^\mathcal{C},\quad (\mathbf{1}^{\mathcal{C}})^\top\bm{A}^{\mathcal{C}}=(\bm{1}^{\mathcal{C}})^\top.
\end{equation} 
The solution to Problem \eqref{Prob:RS2} is given in \citet[Remark 3]{feng2023peer}, depicted as follows:  
\begin{equation*}
\bm{A}^{\mathcal{C}}_*=\frac{\bm{D}(\bm{\gamma}^{\mathcal{C}})^{-1}\bm{1}^{\mathcal{C}}(\bm{1}^{\mathcal{C}})^\top}{\sum_{i\in\mathcal{C}}\gamma_i^{-1}}+k_{\mathcal{C}}\left(\bm{I}_{\mathcal{C}}-\frac{\bm{D}(\bm{\gamma}^{\mathcal{C}})^{-1}\bm{1}^{\mathcal{C}}(\bm{1}^{\mathcal{C}})^\top}{|{\mathcal{C}}|}\right)\bm{\mu}^{\mathcal{C}}(\bm{\mu}^{\mathcal{C}})^\top(\bm{\Sigma}^{\mathcal{C}})^{-1},
\end{equation*}
where $k_{\mathcal{C}}:=((\bm{\mu}^{\mathcal{C}})^\top(\bm{\Sigma}^{\mathcal{C}})^{-1}\bm{\mu}^{\mathcal{C}})^{-1}$. Using this, we define  
\begin{align*}
\mathcal{B}(\mathcal{C}):=\sum_{i\in\mathcal{C}}\left(\mu_i+\frac{\gamma_i\sigma_i^2}{2}\right)-\sum_{i\in\mathcal{C}}\rho_i(\bm{A}^\mathcal{C}_*;\mathcal{C}) \geq0,
\end{align*}
where the inequality follows from the fact that $\bm{I_{\mathcal{C}}}\in\mathbb{R}^{|\mathcal{C}|\times|\mathcal{C}|}$ is also admissible in Problem \eqref{Prob:RS2}. 


\subsubsection{Coalitional stability and the core}

The core \citep{gillies1953some} of the game $(\mathcal{N},\mathcal{B})$ collects the set of allocations such that no subgroup of agents has a joint incentive to deviate from the grand coalition $\mathcal{N}$, which is defined as 
\begin{equation}\label{Eq:CoreNB}
    core(\mathcal{N},\mathcal{B}):=\left\{\bm{c}\in\mathbb{R}_+^{n+1}:\sum_{i\in\mathcal{C}}c_i\geq \mathcal{B}(\mathcal{C}),\ \sum_{i\in\mathcal{N}}c_i= \mathcal{B}(\mathcal{N}),\  \text{for all }  \emptyset\neq\mathcal{C}\subseteq\mathcal{N}\right\}. 
\end{equation}
In other words, if the allocation lies in the core, no subgroup of agents has an incentive to break away from the grand coalition $\mathcal{N}$. The following theorem proves the non-emptiness of the core. 

\begin{theorem}\label{Thm:CoreNB}
The core of the  coalition game  $(\mathcal{N},\mathcal{B})$ is non-empty. 

\end{theorem}

\begin{proof}
    See Appendix \ref{App:CoreNB}
\end{proof}

The next proposition shows how to construct a JP-optimal contract from the core.

\begin{proposition}\label{Prop:coreinJPO}
 Let $\bm{c}\in core(\mathcal{N},\mathcal{B})$, and suppose that $(\bm{A}_*,\bm{p}_*)$ solves Problem \eqref{Prob:RS}. Then, we can uniquely define $\bm{\eta}_*$ such that $\omega_i(\bm{A}_*,\bm{p}_*,\bm{\eta}_*)=c_i$ for all $i\in\mathcal{N}$, and $(\bm{A}_*,\bm{p}_*,\bm{\eta}_*)\in \mathcal{JPO}$.
 
\end{proposition}

\begin{proof}
    See Appendix \ref{App:coreinJPO}.
\end{proof}

The elements in the core lead to stable allocations, which are formally defined as follows.
\begin{definition}[Coalitional stability]\label{Def:CS2}
    A contract $(\bm{A},\bm{p},\bm{\eta})\in\mathcal{IR}$ is called coalitionally stable if there does not exist a subset $\mathcal{C}\subseteq \mathcal{N}$ and $(\bm{A}^\mathcal{C},\bm{p}^\mathcal{C},\bm{\eta}^\mathcal{C})\in\mathcal{F}^\mathcal{C}$ such that for all $i\in\mathcal{C}$, $u_i(\bm{A}^\mathcal{C},\bm{p}^\mathcal{C},\bm{\eta}^\mathcal{C};{\mathcal{C}})\leq u_i(\bm{A},\bm{p},\bm{\eta};{\mathcal{N}})$, and if $R\in\mathcal{C}$, $v(\bm{\eta}^\mathcal{C},\bm{p}^\mathcal{C};{\mathcal{C}})\leq v(\bm{\eta},\bm{p};{\mathcal{N}})$,
 with at least one strict inequality.
\end{definition}
 It is clear that a coalitionally stable contract is JP-optimal.  
The following proposition asserts that JP-optimal contracts with welfare gains allocation lying in the core are coalitionally stable.

\begin{proposition}\label{Prop:CoreiffCS}
   A contract $(\bm{A},\bm{p},\bm{\eta})$ is coalitionally stable  if $(\bm{A},\bm{p},\bm{\eta})\in\mathcal{JPO}$ and 
    \begin{equation*}
\left(\omega_1(\bm{A},\bm{p},\bm{\eta}),\dots,\omega_n(\bm{A},\bm{p},\bm{\eta}),\mathcal{B}(\mathcal{N})-\sum_{i=1}^n\omega_i(\bm{A},\bm{p},\bm{\eta})\right)\in core(\mathcal{N},b). \end{equation*}
\end{proposition}

\begin{proof}
    See Appendix \ref{App:CoreiffCS}
\end{proof}

\subsubsection{Practical premium suggestions}

In practice, the loading under the expected value premium principle is positive and is used to build up risk capital and compensate for administrative costs. Hereafter, we  impose a non-negativity condition $\eta_i\geq 0$, $i=1,\dots,n$, to maintain the closure of feasible contract sets.

{To ensure the reinsurer's IR constraint is met, consider the following price of the reinsurance contracts:}
\begin{equation}\label{Eq:mineta}
    \bm{\eta}_{\min}(\bm{p}):=\frac{\gamma_R}{2}\bm{D}(\bm{\mu})^{-1}\bm{\Sigma}\bm{p}.
\end{equation}
It is clear that \eqref{Eq:ReinsurerIR} is met for any $\bm{\eta} \geq \bm{\eta}_{\min}(\bm{p})$ provided that  {$\bm{p}\geq \bm{0}$}. Combining with the non-negativity requirement, we impose the following minimal safety loadings given by $\max\{\bm{0},\bm{\eta}_{\min}(\bm{p})\}$.\footnote{For $\bm{a},\bm{b}\in\mathbb{R}^n$, $\max\{\bm{a},\bm{b}\}=(\max(a_i,b_i))_{i=1,\ldots,n}$.}

{The following statement provides sufficient conditions for the construction of JP-optimal contracts that satisfy the minimal charge requirement, and establishes a correspondence between core elements and such contracts. The proof of the latter statement relies on the fact that no member can obtain a welfare gain exceeding their marginal contribution, which ensures stability.}

\begin{proposition}[Nonnegative loading in the core]\label{Prop:nonNegEta}
    Let {$(\bm{A}_*,\bm{p}_*)$} be the solution to Problem \eqref{Prob:RS}, 
    and assume that Condition \eqref{Eq:unicond2}  and the following hold for all $i=1,\dots,n$:
    \begin{equation}\label{Eq:WGcond}
    \gamma_i\left[\sigma_i^2 - \left(\frac{\gamma_i^{-1}}{\sum_{j=1}^n\gamma_j^{-1}}\right)^2\sum_{j,l=1}^n (\sigma_{jl})_+ - k  \mu_i^2   \right]\geq\gamma_R\sum_{m=1}^n (\sigma_{im})_+.
    \end{equation}
    {Then, there exists $\bm{\eta}\geq \max\{\bm{0},\bm{\eta}_{\min}(\bm{p}_*)\}$ such that $(\bm{A}_*,\bm{p}_*,\bm{\eta}) \in \mathcal{JPO}$. }
    
    Furthermore, suppose that for any $i=1,\dots,n,$
    \begin{equation}\label{Eq:coreBound}
\mathcal{B}(\mathcal{N})-\mathcal{B}(\mathcal{N}\backslash\{i\})\leq {\omega_i(\bm{A}_*,\bm{p}_*,\max\{\bm{0},\bm{\eta}_{\min}(\bm{p}_*)\}).}
\end{equation}
Then, for any $\bm{c}\in core(\mathcal{N},\mathcal{B})$, there exists $\bm{\eta}_{\bm{c}}\geq \max\{\bm{0},\bm{\eta}_{\min}(\bm{p}_*)\}$ such that $\omega_i(\bm{A}_*,\bm{p}_*,\bm{\eta}_{\bm{c}})=c_i$ for all $i\in\mathcal{N}$ and $(\bm{A}_*,\bm{p}_*,\bm{\eta}_{\bm{c}})\in \mathcal{JPO}$.

\end{proposition}

\begin{proof}
    See Appendix \ref{App:nonNegEta}.
\end{proof}


{Condition \eqref{Eq:WGcond} holds when the risk aversion of the reinsurer is sufficiently small, and risks are mildly correlated. In particular, for a given $\bm{p}$, the minimum safety loading $\bm{\eta}_{\min}(\bm{p})$ required to satisfy the reinsurer’s IR constraint decreases as $\gamma_R$ decreases, which in turn raises the members’ welfare gains {(see \eqref{Eq:wi})} and makes it easier to satisfy their IR constraints.}
Moreover, Condition \eqref{Eq:coreBound} further indicates that all members shall receive welfare gains that exceed their marginal contributions to the pool so that stability is ensured.

The following corollary describes a condition under which the reinsurer can set a single nonnegative loading without violating the stability requirement when she is not allowed to charge differently across members in the pool, {i.e., she must set $\bm{\eta} = t \bm{1}$ for some $t\geq 0$. }

\begin{corollary}[Single loading in the core]\label{Cor:JPOSingleLoading}
    Assume that $core(\mathcal{N},\mathcal{B})$ is not a singleton and that Conditions \eqref{Eq:unicond2}, \eqref{Eq:WGcond} and \eqref{Eq:coreBound} hold.
    Let $\bm{c^{(1)}},\bm{c^{(2)}}\in core(\mathcal{N},\mathcal{B})$,
    and $\bm{\eta^{(1)}}$ and $\bm{\eta^{(2)}}$ the corresponding loadings (see Proposition \ref{Prop:coreinJPO}). If $\min_{i}\eta^{(1)}_i\geq \max_{i}\eta^{(2)}_i$, then
    for any  $t\in [\max_{i}\eta^{(1)}_i, \min_{i}\eta^{(2)}_i]$, $(\bm{A}_*,\bm{p}_*,t\bm{1} )\in\mathcal{JPO}$, and 
    \begin{equation*}
\left(\omega_1(\bm{A}_*,\bm{p}_*,t\bm{1}),\dots,\omega_n(\bm{A}_*,\bm{p}_*,t\bm{1}),\mathcal{B}(\mathcal{N})-\sum_{i=1}^n\omega_i(\bm{A},\bm{p},t\bm{1})\right)\in core(\mathcal{N},b).
    \end{equation*}
\end{corollary}

\begin{proof}
    See Appendix \ref{App:JPOSingleLoading}.
\end{proof}




\section{Bowley Design}\label{Sec:BowleyDesign}

In this section, we examine the risk-sharing scheme of the Bowley design. Under this setting, the reinsurer acts as the leader and has the first-mover advantage of choosing the safety loading $\bm{\eta}$ while anticipating the plan manager’s response. The plan manager, in turn, acts as the follower and determines the risk-mutualization–reinsurance strategy $(\bm{A}(\bm{\eta}), \bm{p}(\bm{\eta}))$ based on the quoted loading $\bm{\eta}$. This leads to two sequentially linked sub-problems, as illustrated in Fig. \ref{fig:Bowley}. 

\begin{figure}[ht]
    \centering

\tikzset{every picture/.style={line width=0.7pt}} 

\begin{tikzpicture}[x=0.7pt,y=0.7pt,yscale=-0.9,xscale=0.9]

\draw   (386,-30) -- (651,-30) -- (651,46) -- (386,46) -- cycle ;
\draw   (5,-30) -- (250,-30) -- (250,46) -- (5,46) -- cycle ;

\draw    (249,-10) -- (384,-10) ;
\draw [shift={(386,-10)}, rotate = 180.42] [color={rgb, 255:red, 0; green, 0; blue, 0 }  ][line width=0.75]    (10.93,-3.29) .. controls (6.95,-1.4) and (3.31,-0.3) .. (0,0) .. controls (3.31,0.3) and (6.95,1.4) .. (10.93,3.29)   ;
\draw    (386,30) -- (251,30) ;
\draw [shift={(249,30)}, rotate = 0.42] [color={rgb, 255:red, 0; green, 0; blue, 0 }  ][line width=0.75]    (10.93,-3.29) .. controls (6.95,-1.4) and (3.31,-0.3) .. (0,0) .. controls (3.31,0.3) and (6.95,1.4) .. (10.93,3.29)   ;

\draw (395,0) node [anchor=north west][inner sep=0.75pt]   [align=left] {P2P Plan Manager ($\bm{A}(\bm{\eta}),\bm{p}(\bm{\eta})$)};
\draw (70,0) node [anchor=north west][inner sep=0.75pt]   [align=left] {Reinsurer ($\bm{\eta}$)};
\draw (18,-70) node [anchor=north west][inner sep=0.75pt]    {Risk-aware profit optimization};
\draw (370,-70) node [anchor=north west][inner sep=0.75pt]    {Pareto allocation and reinsurance game};
\draw (312,-30) node [anchor=north west][inner sep=0.75pt]    {$\bm{\eta}$};
\draw (304,038) node [anchor=north west][inner sep=0.75pt]    {$\bm{p}(\bm{\eta})$};
\draw (95,-90) node [anchor=north west][inner sep=0.75pt]    {Stage 1};
\draw (490,-90) node [anchor=north west][inner sep=0.75pt]    {Stage 2};

\end{tikzpicture}
    \caption{Summary of the sequential game between the reinsurer and the P2P insurance plan manager. The reinsurer selects the risk loading $\eta$, and the insurer the proportional reinsurance strategy $\bm{p}(\bm{\eta})$.} 
    \label{fig:Bowley}
\end{figure}
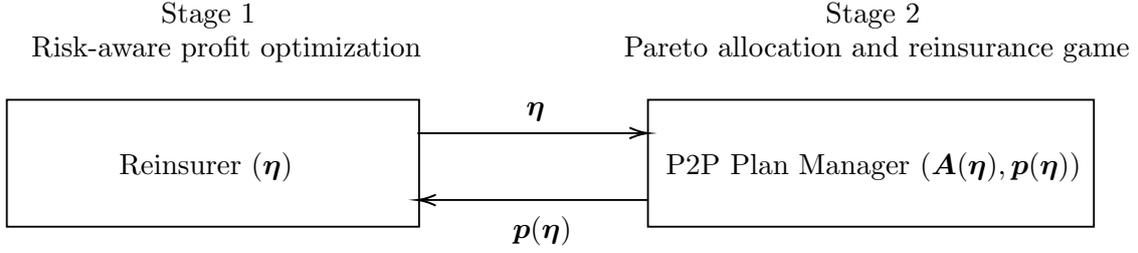


\subsection{Problem formulation}

We formulate sub-problems in the sequential game and thereby define the Bowley-optimal contract.

\subsubsection{Plan manager's follower problem}
We first consider the plan manager's problem. To this end, we assume that the safety loading $\bm{\eta}$ is exogenous and define the manager's optimization problem in Definition \ref{Def:PO}. {To maintain mathematical tractability and obtain closed-form solutions, we initially omit the IR constraints from the follower’s and leader’s problems, and defer the analysis of the IR constraints to Lemma~\ref{Lem:reinsurerIR} and Section~\ref{sec:discussion:eta:bound}.
 }

\begin{definition}[Pareto-optimal risk allocation]\label{Def:PO}
    Given a safety loading $\bm{\eta}$, a risk allocation $(\bm{A}(\bm{\eta})$,$\bm{p}(\bm{\eta}))$ is called Pareto optimal if there exists no other risk allocation $(\bm{\Tilde{A}}(\bm{\eta}),\bm{\Tilde{p}}(\bm{\eta}))$ such that $u_i(\bm{\Tilde{A}}(\bm{\eta}),\bm{\Tilde{p}}(\bm{\eta}),\bm{\eta})\leq u_i(\bm{A},\bm{p},\bm{\eta})$
for all $i=1,\dots,n$, with at least one strict inequality.
\end{definition}

By the standard argument \citep[e.g.,][Chapter 12]{aubin2002optima}, a Pareto-optimal risk allocation is characterized by the following convex minimization problem with the actuarial fairness and zero-conserving constraints:
\begin{equation}\label{Prob:follower}
\min_{\bm{A}(\bm{\eta}),\bm{p}(\bm{\eta})}u(\bm{A}(\bm{\eta}),\bm{p}(\bm{\eta}),\bm{\eta})
        \quad\text{s.t. }\bm{A}(\bm{\eta})\bm{\mu}+\bm{D}(\bm{\mu})\bm{p}(\bm{\eta})=\bm{\mu},\quad \mathbf{1}^\top\bm{A}(\bm{\eta})+\bm{p}^\top(\bm{\eta})=\bm{1}^\top.
\end{equation}

\subsubsection{Reinsurer's leader problem}
Next, we consider the reinsurer's optimization problem. 
{Exploiting her first-mover advantage and taking into account the induced follower response $(\bm{A}^*(\bm{\eta}), \bm{p}^*(\bm{\eta}))$ in Proposition \ref{Prop:convexfollower}}, the reinsurer chooses the loading to optimize her objective, which is given by 
\begin{equation}\label{Prob:leader}
\min_{\bm{\eta}}v(\bm{\eta},\bm{p}^*(\bm{\eta}))\quad
        \text{s.t. }\bm{p}^*(\bm{\eta})\text{ is the optimal indemnity of Problem \eqref{Prob:follower}}.
\end{equation}

Based on the two sub-problems, we define the Bowley-optimal contract as follows. 

\begin{definition}\label{Def:BO}
    A contract $(\bm{A}^*(\bm{\eta}^*),\bm{p}^*(\bm{\eta}^*),\bm{\eta}^*)$ is called Bowley-optimal if it satisfies the following conditions:
    \begin{enumerate}
        \item $(\bm{A}^*(\bm{\eta}^*),\bm{p}^*(\bm{\eta}^*))$ is a Pareto-optimal risk allocation (see Definition \ref{Def:PO});
        \item $\bm{\eta}^*$ solves Problem \eqref{Prob:leader};
        \item $(\bm{A}^*(\bm{\eta}^*),\bm{p}^*(\bm{\eta}^*),\bm{\eta}^*)\in\mathcal{IR}$.
    \end{enumerate}
\end{definition}

\subsection{Solution to the {plan manager's} follower problem}

Following the backward induction, we first tackle the follower's problem. 
Similar to Lemma \ref{Lem:Mbar}, the solution of Problem \eqref{Prob:follower} is characterized by a matrix $\bm{M}$ defined below, whose properties can be proven in the same manner as Lemma \ref{Lem:Mbar}.
    \begin{lemma}
    \label{lem:M}
        The matrix $\bm{M}:=\overline{\bm{M}} - \gamma_R\bm{\Sigma}$
is invertible and positive definite. 
    \end{lemma}

The following proposition solves Problem \eqref{Prob:follower} by the Lagrangian method. 

\begin{proposition}\label{Prop:convexfollower}
   
    The solution to Problem \eqref{Prob:follower} is given by 
    \begin{align}\label{Eq:OptAP}
    \begin{split}
        \bm{A}^*(\bm{\eta})=&\frac{\bm{D}(\bm{\gamma})^{-1}\bm{1}\bm{1}^\top}{\sum_{j=1}^n\gamma_j^{-1}}\bm{D}(\bm{1}-\bm{p}^*(\bm{\eta}))+k\left(\bm{I}_n-\frac{\bm{D}(\bm{\gamma})^{-1}\bm{1}\bm{1}^\top}{\sum_{j=1}^n\gamma_j^{-1}}\right)\bm{D}(\bm{1}-\bm{p}^*(\bm{\eta}))\bm{\mu}\bm{\mu}^\top\bm{\Sigma}^{-1},\\
    \bm{p}^*(\bm{\eta})=&\ \bm{1}-\bm{M}^{-1}\bm{D}(\bm{\mu})\bm{\eta}.
    \end{split}
    \end{align}
In addition, if     \begin{equation}\label{Eq:unicond}
        \begin{aligned}
           { \frac{\sum_{m\neq i}(\sigma_{im}-k\mu_i\mu_m)_+}{\sum_{j=1}^n\gamma_j^{-1}}<\mu_i\eta_i <-\frac{\sum_{m\neq i}(k\mu_i\mu_m-\sigma_{im})_+}{\sum_{j=1}^n\gamma_j^{-1}}+k\mu_i^2\gamma_i+\frac{\sigma_i^2-k\mu_i^2}{\sum_{j=1}^n\gamma_j^{-1}},}
        \end{aligned}
    \end{equation}
    then we have $\bm{p}^*(\bm{\eta})\in(0,1)^n$.
\end{proposition}

\begin{proof}
It follows from an argument similar to the proof of Proposition \ref{Prop:ConvexPareto} in Appendix \ref{App:ConvexPareto}.
\end{proof}

Examining the form of $\bm{p}^*(\bm{\eta})$ and Condition \eqref{Eq:unicond} entails a bound on the loading factor $\bm{\eta}$. Economically, this implies that {reinsurance contracts must be priced within a moderate range: excessively high loadings lead to zero reinsurance, while overly low loadings induce full risk transfer.}


\subsection{Solution to the {reinsurer's} leader problem}

We then solve the leader's problem.
Using Proposition \ref{Prop:convexfollower}, Problem \eqref{Prob:leader} can be transformed into:
\begin{equation}\label{Prob:leader2}
\min_{\bm{\eta}}\frac{\gamma_{R}}{2}\bm{p}^\top\bm{\Sigma}\bm{p}-(\bm{D}(\bm{\mu})\bm{\eta})^\top\bm{p}\quad
        \text{s.t. }       \bm{p}=\bm{1}-\bm{M}^{-1}\bm{D}(\bm{\mu})\bm{\eta}. 
\end{equation} 
The following proposition gives the optimal loading factor to Problem \eqref{Prob:leader2}.

\begin{proposition}\label{Prop:leader} 
    The optimal safety loading of Problem \eqref{Prob:leader2} is
    \begin{equation}\label{Eq:Opteta}
        \bm{\eta}^*:=\bm{D}(\bm{\mu})^{-1}\bm{M}(\gamma_{R}\bm{\Sigma}+2\bm{M})^{-1}(\gamma_{R}\bm{\Sigma}+\bm{M})\bm{1}.
    \end{equation}
    {Moreover, suppose that $\bm{M}$ and $\bm{\Sigma}$ are strictly diagonally dominant, i.e., $\delta_{\bm{M}}:=\min_{i\in\{1,\ldots,n\}}(M_{ii}-\sum_{j\neq i}|M_{ij}|)>0$,  $\delta_{\bm{\Sigma}}:=\min_{i\in\{1,\ldots,n\}}(\sigma_{ii}-\sum_{j\neq i}|\sigma_{ij}|)>0$,  where $M_{ij}$ denotes the $(i,j)$-th entry of $\bm{M}$, and 
     \begin{equation}\label{Eq:deltaINE}
        \delta_{\bm{M}}+\frac{\gamma_R}{2}\delta_{\bm{\Sigma}}\geq \frac{\gamma_R^2\|\bm{\Sigma}\|_{\infty}^2}{2(2\delta_{\bm{M}}+\gamma_R\delta_{\bm{\Sigma}})},
    \end{equation}
$\|\cdot\|_\infty$ denotes the infinity-norm of a square matrix. Then, $\bm{\eta}^*\geq\bm{0}$.
    }
\end{proposition}

\begin{proof}
    See Appendix \ref{App:leader}.
\end{proof}
\begin{remark}
    The condition \eqref{Eq:deltaINE} implies an entry-wise lower bound on the optimal safety loading $\bm{\eta}^*$. {In Appendix~\ref{sec:discussion:eta:bound}, we invoke the idea of the proof of the non-negativity of $\bm{\eta}^*$ and establish an upper bound on $\bm{\eta}^*$. This also provides a way to verify \eqref{Eq:unicond} under the reinsurer's optimal strategy, which requires substituting $\bm{\eta}^*$ into the constraints. In particular, when $\gamma_R=0$, \eqref{Eq:unicond} is equivalent to, for $i=1,\dots,n$,  
    \begin{equation}
    \label{Eq:unicond2:gam_R=0}
            k\mu_i^2\gamma_i + \frac{\sigma_i^2-k\mu_i^2}{\sum_{j=1}^n\gamma_j^{-1}} > \frac{\sum_{j\neq i}|\sigma_{ij}-k\mu_i\mu_j|}{\sum_{j=1}^n\gamma_j^{-1}} . 
    \end{equation}
   }%
\end{remark}

\subsection{Bowley optimum}

With the optimal risk allocation in Proposition \ref{Prop:convexfollower} and the optimal loading in Proposition \ref{Prop:leader}, we have a candidate for the desired Bowley optimum. 
Following Definition \ref{Def:BO}, the next step is to check the individual rationality of all agents, shown in the following statement.



\begin{lemma}\label{Lem:reinsurerIR}
  Under the contract $(\bm{A}^*(\bm{\eta}^*),\bm{p}^*(\bm{\eta}^*),\bm{\eta}^*)$ given by Propositions \ref{Prop:convexfollower} and \ref{Prop:leader}, the reinsurer's welfare gain is given by 
    \begin{equation*}
     \omega_R(\bm{A}^*(\bm{\eta}^*),\bm{p}^*(\bm{\eta}^*),\bm{\eta}^*)  = \frac{1}{2}\bm{1}^\top \left(\gamma_{R}\bm{M}^{-1}\bm{\Sigma}\bm{M}^{-1}+2\bm{M}^{-1} \right)^{-1} \bm{1}>0. 
    \end{equation*}
In addition, the members' IRs are fulfilled provided that $\bm{\eta}^*$ satisfies \eqref{Eq:unicond}, and the following holds for $i=1,\dots,n$:
    \begin{equation}\label{Eq:MIRcond}
    \begin{aligned}
             &  \ \ \ \    \frac{\gamma_i}{2}\left(\sigma_i^2 - \left(\frac{\gamma_i^{-1}}{\sum_{j=1}^n\gamma_j^{-1}}\right)^2\sum_{j,l=1}^n (\sigma_{jl})_+ - 3k  \mu_i^2   \right)   + \frac{\sum_{m\neq i} (k\mu_i\mu_m-\sigma_{im})_+   + k\mu_i^2 - \sigma_i^2 }{\sum_{j=1}^n\gamma_j^{-1}}  \geq 0. 
    \end{aligned}
    \end{equation}
\end{lemma}
\begin{proof}
    See Appendix \ref{App:reinsurerIR}.
\end{proof}

Synthesizing the solutions to two sequentially-linked sub-problems and the IR analysis, we characterize the Bowley optimum as follows. 


\begin{theorem}\label{Thm:BO}
    {Suppose {that C}onditions \eqref{Eq:unicond} and \eqref{Eq:MIRcond} are fulfilled. Let $(\bm{A}^*(\bm{\eta}^*),\bm{p}^*(\bm{\eta}^*),\bm{\eta}^*)$ be given in \eqref{Eq:OptAP} and \eqref{Eq:Opteta}. Then, $(\bm{A}^*(\bm{\eta}^*),\bm{p}^*(\bm{\eta}^*),\bm{\eta}^*)$ is Bowley-optimal.}
\end{theorem} 

\begin{proof}
    It follows directly from Proposition \ref{Prop:convexfollower}, Proposition \ref{Prop:leader}, and Lemma \ref{Lem:reinsurerIR}.
\end{proof}

The following corollary characterizes the optimal loading when the reinsurer is not allowed to set different prices for members in the pool.

\begin{corollary}\label{Lem:singleprice}
    With {the single-loading restriction}, i.e., $\bm{\eta}=\eta\bm{1}$, 
    the optimal safety loading of Problem \eqref{Prob:leader2} is given by 
    \begin{equation}\label{Eq:OptSingleeta}
        \eta^*=\frac{\bm{\mu}^\top(\bm{I}_n+\gamma_{R}\bm{M}^{-1}\bm{\Sigma})\bm{1}}{\bm{\mu}^\top(2\bm{M}^{-1}+\gamma_{R}\bm{M}^{-1}\bm{\Sigma}\bm{M}^{-1})\bm{\mu}}.
    \end{equation}
    Assume further that Conditions \eqref{Eq:unicond} and \eqref{Eq:MIRcond} are fulfilled; then $(\bm{A}^*(\eta^*\bm{1}),\bm{p}^*(\eta^*\bm{1}),\eta^*\bm{1})$ is the Bowley-optimal contract. 
    In addition, if $\bm{M}$ is strictly diagonally dominant and $\gamma_R\leq\delta_{\bm{M}}/\|\bm{\Sigma}\|_{\infty}$, 
    then $\eta^*\geq 0$. 
\end{corollary}
\begin{proof}
    See Appendix \ref{App:singleprice}.
\end{proof}


\begin{remark}[On including individual rationality constraints in optimizations]
    For mathematical tractability, we do not directly include two IR constraints \eqref{Eq:MemberIR} and \eqref{Eq:ReinsurerIR} in the optimization problems \eqref{Prob:follower} and \eqref{Prob:leader}, respectively. We nonetheless note that both conditions can be added to the leader's problem \eqref{Prob:leader}, yielding a convex minimization problem that is solvable. However, a closed-form solution is not available in this case, and one must rely on numerical optimization solvers. 
\end{remark}

\subsection{Comparison of Bowley and Pareto optima}

We provide a brief comparison between the Bowley-optimal contract $(\bm{A}^*(\bm{\eta}^*),\bm{p}^*(\bm{\eta}^*),\bm{\eta}^*)$ and a JP-optimal contract $(\bm{A}_*,\bm{p}_*,\bm{\eta}_*)$. Despite the flexibility of setting $\bm{\eta}_*$ in the JP-optimal contracts (see the discussion in Section \ref{Sec:CJPO}),   
we shall show that a Bowley-optimal contract is never JP-optimal.   

\begin{proposition}\label{Prop:BOisJPO}
  The Bowley-optimal contract $(\bm{A}^*(\bm{\eta}^*),\bm{p}^*(\bm{\eta}^*),\bm{\eta}^*)$ in Theorem \ref{Thm:BO} is not JP-optimal.

    
\end{proposition}

\begin{proof}
    By direct calculation, $\bm{p}^*(\bm{\eta}^*)=(\gamma_R\bm{\Sigma}+2\bm{M})^{-1}\bm{M}\bm{1}\neq (\gamma_R\bm{\Sigma}+\bm{M})^{-1}\bm{M}\bm{1}=\bm{p}_*$, 
    as $\bm{M}$ is not a zero matrix according to Lemma \ref{lem:M}. {As $\bm{p}^*(\bm{\eta}^*){\neq}\bm{p}_*$, the Bowley-optimal contract does not solve Problem \eqref{Prob:RS} and thus not JP-optimal according to Theorem \ref{Thm:JPO}, } 
\end{proof}



Proposition \ref{Prop:BOisJPO} asserts that the Bowley optimum is inefficient in terms of JP optimality. 
This stands in contrast to \cite{boonen2023bowley}, who study a two-agent Bowley optimum under convex and comonotonic additive risk measures. In contrast, our focus is on a multi-agent problem with mean-variance-premium objectives for the followers, which are non-convex in the contract variables and not comonotonic-additive.  Our result echoes a similar finding in \cite{jiang2025bowley} that the two-agent Bowley optimum is not Pareto efficient under the generalized mean-variance setting. 

In addition, due to the inefficiency of the Bowley optimum, we can also assert that it achieves smaller total welfare gains for all agents compared to the JP-optimal counterpart, summarized in the following corollary.

\begin{corollary}\label{Cor:WelfareGain} 

{Let $(A^*(\bm{\eta}^*),\bm{p}^*(\bm{\eta}^*),\bm{\eta}^*)$ be the Bowley-optimal contract given by \eqref{Eq:OptAP} and \eqref{Eq:etastar}, and {let} $(\bm{A}_*,\bm{p}_*,\bm{\eta}_*)$ {be a} JP-optimal contract. Then, $\sum_{i\in\mathcal{N}}\omega_i(\bm{A}^*(\bm{\eta}^*),\bm{p}^*(\bm{\eta}^*),\bm{\eta}^*)<\sum_{i\in\mathcal{N}}\omega_i(\bm{A}_*,\bm{p}_*,\bm{\eta}_*)$.}
\end{corollary}

\begin{proof}
    It follows directly from the observation that $(\bm{A}^*,\bm{p}^*)$ is admissible and suboptimal (see Proposition \ref{Prop:BOisJPO}) in Problem \eqref{Prob:RS}. 
\end{proof}

Overall, the issues of inefficiency and welfare loss, respectively revealed in Proposition \ref{Prop:BOisJPO} and Corollary \ref{Cor:WelfareGain}, raise concerns about adopting the Bowley design with the monopolistic reinsurer. Importantly, 
the implication from Corollary \ref{Cor:WelfareGain}
should not be interpreted as a blanket rejection of Bowley-type contracting as it is more realistic in the practical business models. In our setting, the Pareto design corresponds to full cooperation and thus serves as an upper benchmark, whereas the Bowley design represents the opposite extreme of unilateral pricing power. This motivates further study of intermediate market structures—for example, Nash-bargaining-based contract composition \citep{boonen2016nash} and structures with multiple reinsurers \citep{cao2023reinsurance}.

\section{Numerical Illustration and Welfare Analysis}\label{Sec:Numeric}



In this section, we present numerical results to demonstrate our findings and conduct a welfare comparison across different contracts. In addition, we carry out a comparative statics analysis on the reinsurer's risk aversion $\gamma_R$.
Baseline parameters are shown as follows: we consider a three-member case, where
\begin{equation*}
    \mathcal{N}=\{1,2,3,R\},\quad\gamma_R=0.01,\quad\bm{\mu}=\begin{pmatrix}
        100\\
        125\\
        85
    \end{pmatrix}
    ,\quad
    \bm{\Sigma}=\begin{pmatrix}
        10000& -1200& 720\\
    -1200& 14400& 648\\
    720& 648& 8100
    \end{pmatrix}
    ,\quad
    \bm{\gamma}=\begin{pmatrix}
        0.015\\ 0.025\\ 0.02
    \end{pmatrix}.
\end{equation*}

{In this profile, Member 2 faces the riskiest loss, with the highest expected loss and variance, whereas Member 3 faces the least risky loss. Among the three members, Member 1 is the least risk averse, followed by Member 3.}

\subsection{Baseline contracts}

From the baseline setting, we compose the following five contracts:
\begin{itemize}
\item No-reinsurer case $\bm{A}_0$
\item JPO1 $(\bm{A}_*,\bm{p}_*,\bm{\eta}_*)$: JP-optimal contract 
\item JPO2 $(\bm{A}_*,\bm{p}_*,\eta_*\bm{1})$: JP-optimal contract {with the single-loading restriction}
\item BO1 $(\bm{A}^*(\bm{\eta}^*),\bm{p}^*(\bm{\eta}^*),\bm{\eta}^*)$: Bowley-optimal contract 
\item BO2 $(\bm{A}^*(\eta^*\bm{1}),\bm{p}^*(\eta^*\bm{1}),\eta^*\bm{1})$: Bowley-optimal contract {with the single-loading restriction}

\end{itemize}
Below, we discuss each component in the contract triplet and related figures.

\subsubsection{Construction of JP-optimal contracts and safety loadings}\label{Sec:SLs}




Safety loadings of the Bowley-optimal contracts can be calculated explicitly; see \eqref{Eq:Opteta} and \eqref{Eq:OptSingleeta} for the closed-form formulae. However, this is not the case for Pareto designs, which allow the plan manager and reinsurer to jointly devise the safety loadings. As such, we select the premiums for the JPO1 and JPO2 contracts so that they respectively Pareto-dominate the BO1 and BO2 contracts, as follows.

According to Corollary \ref{Cor:WelfareGain}, JP-optimal contracts always result in greater total welfare gains for the community. Utilizing the freedom of distributing the total welfare among members and the reinsurer in the Pareto design, we equally redistribute these extra welfare gains among them and derive $\bm{\eta}_*$ in the JPO1 contract. Mathematically, {whenever $\bm{p}_*\geq \bm{0}$,} we {determine $\bm{\eta}_*$ by solving}
\begin{align}
\label{eq:omega:dis:numerical}
    &\omega_i(\bm{A}_*,\bm{p}_*,\bm{\eta}_*) = \omega_i(\bm{A}^*(\bm{\eta}^*),\bm{p}^*(\bm{\eta}^*),\bm{\eta}^*)+\frac{\rho_R(\bm{p}_*)-\rho_R(\bm{p}^*(\bm{\eta}^*))+\sum_{i=1}^n(\rho_i(\bm{A}_*)-\rho_i(\bm{A}^*(\bm{\eta}^*)))}{|\mathcal{N}|},
\end{align}
for all $i\in\mathcal{N}$.
We also verify that the resulting vector of welfare gains lies in the core; see Equation \eqref{Eq:CoreNB}. Therefore, the JPO contracts are also coalitionally stable.


For $\eta_*$ in the JPO2 contract, we follow a similar scheme by redistributing the welfare gains. However, due to the single-loading restriction, we cannot, in general, equally split the extra welfare gains among agents while satisfying the coalitional stability criteria. Instead, we search for welfare gains induced by the single loading $\eta_*$ of the JP-optimal contract in the core that Pareto-dominates the BO2 contract. Mathematically, we look for $\eta_*\geq0$ such that
\begin{align*}
\left(\omega_i(\bm{A}_*,\bm{p}_*,\eta_*\bm{1})\right)_{i\in\mathcal{N}}
\in&\{\bm{c}:\bm{c}
\in core(\mathcal{N},\mathcal{B}),\  c_i\geq\omega_i(\bm{A}^*(\eta^*\bm{1}),\bm{p}^*(\eta^*\bm{1}),\eta^*\bm{1})\text{ for all }i\in\mathcal{N}\}.
\end{align*}
Note that  in general, there is no unique solution for this search, and we thus select one particular choice to construct the JPO2 contract. 

\begin{table}[!ht]
\centering
{\small\begin{tabular}{|c|ccc|}
\hline
Contracts\textbackslash{}Safety loading & \multicolumn{1}{c|}{$\eta_1$} & \multicolumn{1}{c|}{$\eta_2$} & $\eta_3$ \\ \hline
JPO1                                    & \multicolumn{1}{c|}{0.313589}  & \multicolumn{1}{c|}{0.678401}  & 0.404502  \\ \hline
JPO2                                    & \multicolumn{3}{c|}{0.4395}                                                \\ \hline
BO1                                     & \multicolumn{1}{c|}{0.345775}  & \multicolumn{1}{c|}{0.725918}  & 0.460025  \\ \hline
BO2                                     & \multicolumn{3}{c|}{0.495050}                                             \\ \hline
\end{tabular}}
\caption{Safety loadings under the baseline contracts.}
\label{Tab:SafetyLoadings}
\end{table}

The computed safety loadings of all contracts are summarized in Table \ref{Tab:SafetyLoadings}. All loadings are positive, and the loadings of the two contracts (JPO2/BO2) with the single-loading restriction are located around the average of safety loadings in the corresponding unrestricted contracts (JPO1/BO1).  In line with general intuition, members with larger expected losses and higher loss variances are charged higher loadings in both the JPO1 and BO1 contracts, with Member 2 facing the highest loading. {In addition, the loading for Member 2 is reduced under the single-loading restriction, which is balanced by higher loadings for Members 1 and 3.} Comparing the corresponding contracts across the two designs, we observe that the higher loadings are set under the Bowley design, as the reinsurer, acting as the leader, can leverage its first-mover advantage when determining the premium.

\subsubsection{Reinsurance strategies}

\begin{table}[!ht]
\centering
{\small\begin{tabular}{|c|c|c|c|}
\hline
Contracts \textbackslash Reinsurance & $p_1$    & $p_2$    & $p_3$    \\ \hline
JPO1 and JPO2                        & 0.357528 & 0.473022 & 0.364981 \\ \hline
BO1                                  & 0.265399 & 0.378127 & 0.269785 \\ \hline
BO2                                  & 0.148915 & 0.445595 & 0.244199 \\ \hline
\end{tabular}}
\caption{Comparison of reinsurance strategies among contracts.}
\label{Tab:Reinsurance_comparison}
\end{table}


Table \ref{Tab:Reinsurance_comparison} depicts all members' reinsurance strategies in different contracts. The reinsurance strategies are the same for JPO1 and JPO2, since under the Pareto design, the risk transfer is independent of the safety loadings. In general, the Pareto design leads to higher reinsurance levels for all members compared to the Bowley design. This is because the reinsurer cooperates with the manager and internalizes members’ risk-sharing benefits, thereby allowing more effective risk management. For all but Member 2, the next highest levels occur under the BO1 contract, since her safety loading decreases substantially when moving from BO1 to BO2, while the other members’ loadings increase (see Table \ref{Tab:SafetyLoadings}). 


\subsubsection{Risk allocation matrices}

\begin{table}[!ht]
\centering

\begin{subtable}[t]{0.495\textwidth}
\centering
{\small\begin{tabular}{|c|c|c|c|}
\hline
$i\backslash j$ & 1 & 2 & 3 \\ \hline
1 & 0.223266 & 0.180385 & 0.227910 \\ \hline
2 & 0.231155 & 0.193278 & 0.218787 \\ \hline
3 & 0.188050 & 0.153314 & 0.188321 \\ \hline
\end{tabular}}%
\caption{JPO1 and JPO2 $\bm{A}_*$}
\end{subtable}\hfill
\begin{subtable}[t]{0.495\textwidth}
\centering
{\small\begin{tabular}{|c|c|c|c|}
\hline
$i\backslash j$ & 1 & 2 & 3 \\ \hline
1 & 0.312006 & 0.326195 & 0.329707 \\ \hline
2 & 0.418419 & 0.298033 & 0.392988 \\ \hline
3 & 0.269576 & 0.275772 & 0.277305 \\ \hline
\end{tabular}}%
\caption{No-reinsurer case $\bm{A}_0$}
\end{subtable}

\vspace{0.5em}
\begin{subtable}[t]{0.495\textwidth}
\centering
{\small\begin{tabular}{|c|c|c|c|}
\hline
$i\backslash j$ & 1 & 2 & 3 \\ \hline
1 & 0.239741 & 0.226409 & 0.249234 \\ \hline
2 & 0.290250 & 0.263952 & 0.273118 \\ \hline
3 & 0.204610 & 0.191511 & 0.207863 \\ \hline
\end{tabular}}%
\caption{BO1 $\bm{A}^*(\bm{\eta}^*)$}
\end{subtable}\hfill
\begin{subtable}[t]{0.495\textwidth}
\centering
{\small\begin{tabular}{|c|c|c|c|}
\hline
$i\backslash j$ & 1 & 2 & 3 \\ \hline
1 & 0.333939 & 0.211219 & 0.297793 \\ \hline
2 & 0.265230 & 0.183491 & 0.233428 \\ \hline
3 & 0.251918 & 0.159696 & 0.224580 \\ \hline
\end{tabular}}%
\caption{BO2 $\bm{A}^*(\eta^*\bm{1})$}
\end{subtable}

\caption{Risk allocations $(a_{ij})_{i,j\in\{1,2,3\}}$ of members in different contracts.
}
\label{tab:A}
\end{table}

Table \ref{tab:A} presents the risk mutualization of members in all contracts. 
Driven by the introduction of the reinsurance option, we observe two key effects. First, {the allocation matrix $\bm{A}_0$ has the highest values in all entries, followed by $\bm{A}^*(\bm{\eta}^*)$ and then $\bm{A}_*$.}
This follows from the ordering of the corresponding reinsurance strategies, i.e., $\bm{0}<\bm{p}^*(\bm{\eta}^*)<\bm{p}_*$ 
(cf. Table \ref{Tab:Reinsurance_comparison}). This is consistent with the intuition that members need to carry more risk from each other with fewer risk transfers. 

For the BO2 contract where the single-loading restriction is imposed, {the entries of} $\bm{A}(\eta^*\bm{1})$ {exhibit greater dispersion} compared to other allocation matrices. {In particular, $a_{11}$ in $\bm{A}(\eta^*\bm{1})$ is the highest among all contracts, whereas $a_{22}$ is the smallest.}
The reasons are twofold. First, the lower premium payment incentivizes Member 2 to transfer a large portion of risk to the reinsurer. Consequently, she bears a smaller proportion of losses for others due to the actuarial fairness condition. As such, Member 1 must retain more risk in her own house. Second, for Member 1, the premium is substantially higher in the BO2
contract than in the other designs, making the reinsurance option less appealing and thereby increasing her retained loss.

\subsubsection{Welfare analysis}


In this subsection, we compare the induced welfare of all agents and analyze the attribution of welfare improvements with respect to the premium and mean-variance disutility components.



\begin{table}[!ht]
\centering
{\small\begin{tabular}{|c|c|c|c|c|}
\hline
Contracts \textbackslash Premium & $\pi_1(\eta_1,p_1)$ & $\pi_2(\eta_2,p_2)$ & $\pi_3(\eta_3,p_3)$ & Total premium payment \\ \hline
JPO1                             & 46.9654             & 99.2401             & 43.5725             & 189.777               \\ \hline
JPO2                             & 51.4662             & 85.1144             & 44.6582             & 181.239               \\ \hline
BO1                              & 35.7167             & 68.6327             & 33.4809             & 137.830               \\ \hline
BO2                              & 22.2633             & 83.2733             & 31.0327             & 136.569               \\ \hline
\end{tabular}}%
\caption{Comparison of premium payment among contracts.}
\label{Tab:premium_comparison}
\end{table}

Table~\ref{Tab:premium_comparison} reports the members' premium payments under different contracts, showing that members pay higher premiums under the JP-optimal contracts. {With lower safety loadings in the JPO contracts (see Table~\ref{Tab:SafetyLoadings}), members are incentivized to increase their reinsurance purchases (see Table~\ref{Tab:Reinsurance_comparison}). Overall, the higher reinsurance volumes offset the lower unit premiums, resulting in higher total premium payments.} 
Consequently, the reinsurer receives higher premium income when engaging in the JP-optimal contracts.

\begin{table}[!ht]
\centering
{\small\begin{tabular}{|c|c|c|c|c|}
\hline
Contracts \textbackslash Mean-variance disutility & $\rho_1(\bm{A},\bm{p})$ & $\rho_2(\bm{A},\bm{p})$ & $\rho_3(\bm{A},\bm{p})$ & $\rho_R(\bm{p})$ \\ \hline
Status quo                                        & 175                     & 305                     & 166                     & 0                \\ \hline
No-reinsurer case                                                & 125.721                 & 191.536                 & 109.730                 & N/A              \\ \hline
JPO1 and JPO2                                     & 74.9797                 & 84.3771                 & 63.9625                 & 153.829          \\ \hline
BO1                                               & 87.2974                 & 116.154                 & 75.2239                 & 103.052          \\ \hline
BO2                                               & 104.094                 & 90.0197                 & 78.6610                 & 109.338          \\ \hline
\end{tabular}}%
\caption{Mean-variance disutilities among agents in different contracts.}
	\label{Tab:MeanVariance}
\end{table}

Table \ref{Tab:MeanVariance} reports the mean-variance disutilities for all agents, which indicate the level of risk borne by each agent.
We note that 
including the reinsurance layer 
yields a {substantially} greater risk reduction for all members than in the case without reinsurance. As the mean-variance disutilities attain their minimum in JP-optimal contracts, this indicates better risk sharing among agents in the full-cooperation framework. 
Juxtaposing the two Bowley-optimal contracts in Table \ref{Tab:MeanVariance} reflects the impact of the single-loading restriction. As Member 2 cedes more risk in the BO2 contract (see Table \ref{Tab:Reinsurance_comparison}), she faces a lower mean-variance disutility. Note that similar arguments can be made to explain why Members 1 and 3 instead bear more risk in the BO2 contract.

\begin{table}[!ht]
\centering
{\small\begin{tabular}{|c|c|c|c|c|c|}
\hline
     & $\omega_1(\bm{A},\bm{p},\bm{\eta})$ & $\omega_2(\bm{A},\bm{p},\bm{\eta})$ & $\omega_3(\bm{A},\bm{p},\bm{\eta})$ & $\omega_R(\bm{A},\bm{p},\bm{\eta})$ & Total welfare gains       \\ \hline
No-reinsurer case   & 49.2790                             & 113.464                             & 56.2696                             & N/A                                 & 219.012                   \\ \hline
JPO1 & 53.1558                             & 121.383                             & 58.4651                             & 35.9478                             & \multirow{2}{*}{268.951} \\ \cline{1-5}
JPO2 & 48.6541                             & 135.508                             & 57.3793                             & 27.4096                             &                           \\ \hline
BO1  & 51.9859                             & 120.213                             & 57.2952                            & 34.7780                            & 264.272                   \\ \hline
BO2  & 48.6430                             & 131.707                             & 56.3063                             & 27.2311                             & 263.888                   \\ \hline
\end{tabular}}%
\caption{Welfare gains among contracts.}
\label{Tab:WelfareGains}
\end{table}

Table \ref{Tab:WelfareGains} reports the overall welfare gains of agents. Scrutinizing the effect of reinsurance, we note that all figures are much lower in the design without the reinsurer than in those with the reinsurer, highlighting the value of the reinsurance option and thus echoing the finding in \citet[Theorem 4]{anthropelos2026expansion} despite a different research focus.

Concerning the effect of the single-loading restriction, under the Bowley design, the total welfare gain of the BO2 contract is smaller than that of the BO1 contract. {Indeed, except for Member 2 who benefits substantially from the reduced loading (see Table \ref{Tab:SafetyLoadings}), all other agents experience a decline in individual welfare. This suggests that the single-loading restriction primarily benefits members with riskier losses, while failing to improve overall welfare and potentially harming it. In particular, Table \ref{Tab:premium_comparison} indicates that the reinsurer's welfare loss is due to a lower premium income.  
A similar effect on the individual welfare is observed when comparing JPO1 and JPO2 contracts, except that the total welfare gain is preserved due to the Pareto design. }

Between the two designs, all members and the reinsurer prefer the Pareto design to the Bowley design, with all individual welfare gains being higher under the JPO1 (resp.~JPO2) contract than under the BO1 (resp.~BO2) contract. The intuition is straightforward: the Pareto design generates a higher total welfare gain by default and, owing to the flexibility in setting the reinsurance premium, the resulting surplus is redistributed among all agents according to \eqref{eq:omega:dis:numerical}. It is also noteworthy that the reinsurer attains a higher welfare gain under the JP-optimal contracts owing to the fact that the higher premium income more than compensates for the additional risk (see Tables \ref{Tab:premium_comparison} and \ref{Tab:MeanVariance}).
From the members' perspective, although they need to pay more under JP-optimal contracts, they are well-compensated with better risk sharing and thus enjoy greater welfare gains.  

Consequently, for all agents except Member 2, individual welfare gains are maximized under the JPO1 contract.
By contrast, benefiting from the reduced loading, Member 2 attains the highest welfare gain under the JPO2 contract. Moreover, the asymmetric allocation of welfare toward Member 2 in the JPO2 contract results in lower individual welfare gains for all other agents compared to the JPO1 contract, even though their total welfare is the same. This observation further reinforces the conclusion that the single-loading restriction disproportionately favors high-risk members.   }



\subsection{Impact of the reinsurer's risk aversion}

In this subsection, we examine the effect of the reinsurer's risk aversion $\gamma_R$ on members' reinsurance policies and the total welfare gains across various contracts. Given this focus, we do not consider the construction of coalitionally stable JP-optimal contracts.

\subsubsection{Members' reinsurance policies}

\begin{figure}[!ht]	
	\begin{subfigure}[t]{.33\linewidth}
		\centering\includegraphics[scale=0.285]{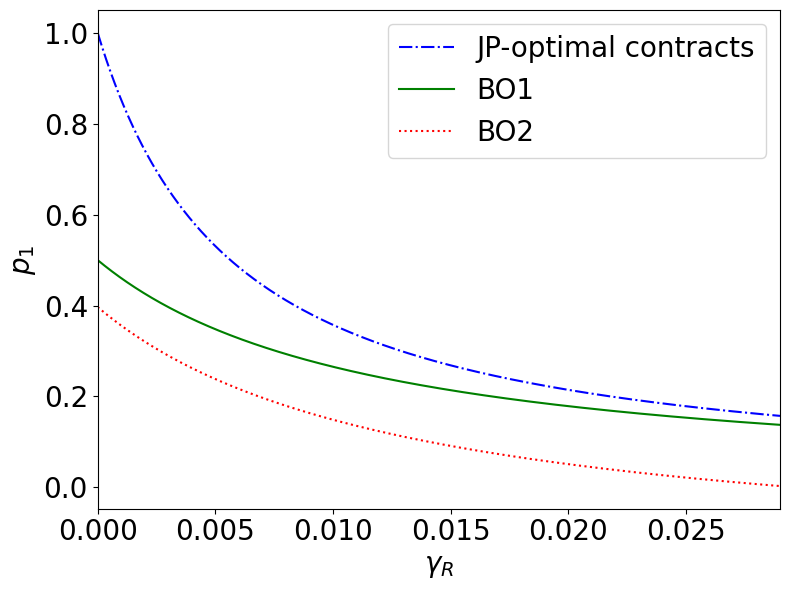}
		\caption{Member 1}
		\label{fig:p1}
	\end{subfigure}%
        \begin{subfigure}[t]{.33\linewidth}
		\centering\includegraphics[scale=0.285]{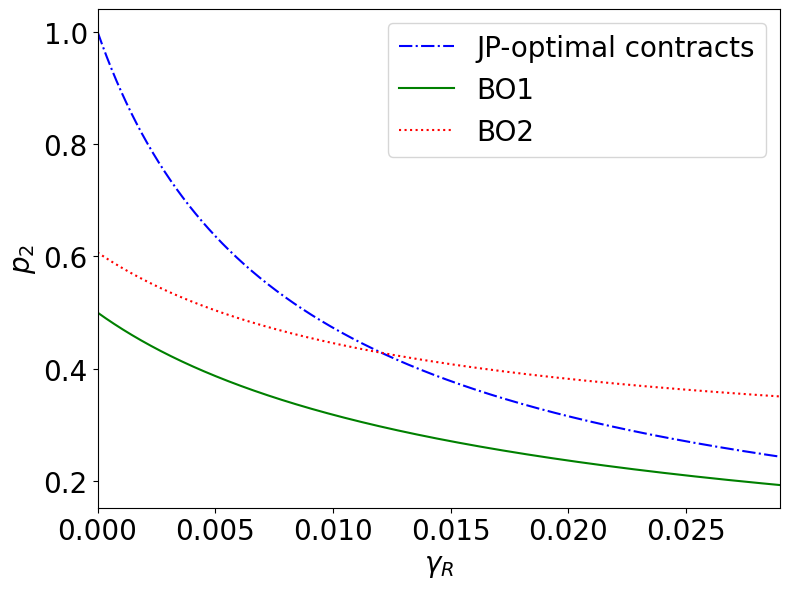}
		\caption{Member 2}
		\label{fig:p2}
	\end{subfigure}%
    \begin{subfigure}[t]{.33\linewidth}
		\centering\includegraphics[scale=0.285]{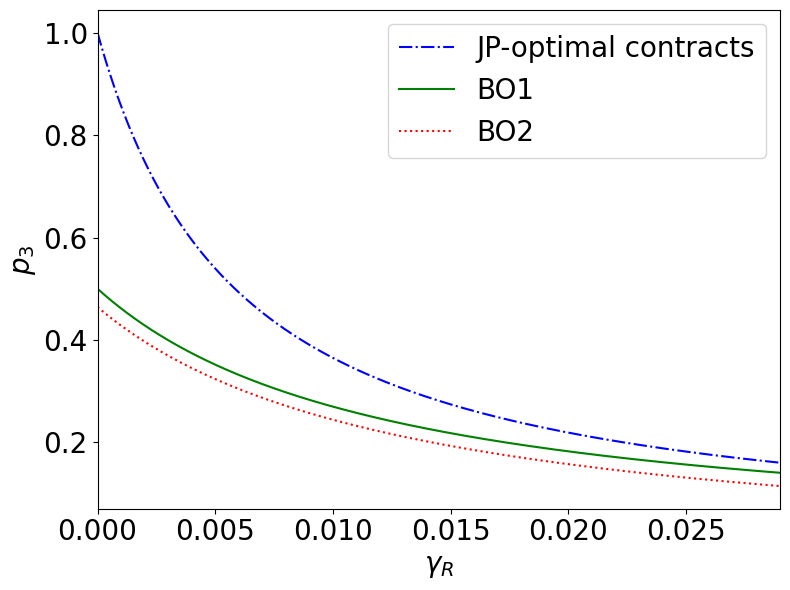}
		\caption{Member 3}
		\label{fig:p3}
	\end{subfigure}
    \caption{Members' reinsurance strategies of the four contracts with different $\gamma_R$.}
	\label{fig:pVSGammaR}
\end{figure}


Fig.~\ref{fig:pVSGammaR} displays the members' reinsurance strategies $p_i$ against $\gamma_R$.
In general, every $p_i$ is decreasing in $\gamma_R$. This is consistent with the general intuition that the more risk averse the reinsurer, the less risk she is willing to take. In addition, when $\gamma_R=0$, i.e., the reinsurer is risk-neutral, all members adopt full reinsurance in the JP-optimal contracts. {This can also be deduced from the form of $\bm{p}_*$ given in \eqref{Eq:OptAP2}.} Moreover, members cede only a portion of their risks when the reinsurer is risk-neutral. In particular, under the BO1 contract, {the reinsurance proportion is always $50\%$ for all members when $\gamma_R=0$, which follows directly from \eqref{Eq:OptAP} and \eqref{Eq:Opteta}}.

{Concerning the impact of the single-loading restriction, we note that Member 1's reinsurance is near zero in the BO2 contract when $\gamma_R=0.029$.\footnote{This also explains the choice of the plotting range.}
This indicates that the single-loading restriction can discourage the least risk-averse member from participating in reinsurance protection as the reinsurer becomes more risk averse. Indeed, as noted earlier, this restriction tends to favor members with riskier losses. When $\gamma_R>0.013$, Member~2 adopts the highest reinsurance level under the BO2 contract, highlighting the extent to which the single-loading restriction enhances her welfare with a more risk-averse reinsurer.}


{Comparing the three contract designs, we find that the ordering of reinsurance proportions in Fig.~\ref{fig:pVSGammaR} largely aligns with that reported in Table~\ref{Tab:Reinsurance_comparison}, with reinsurance levels typically being highest for all members under the JP-optimal contracts. Focusing on the BO1 and JP-optimal contracts, this finding echoes the results of \citet[Theorem~4.3]{jiang2025bowley}, despite the different modeling framework.  The only exception to the above pattern occurs for Member~2, for whom the reinsurance level under the BO2 contract exceeds that under the JP-optimal contracts when $\gamma_R$ becomes sufficiently large for the aforementioned reason. In addition, the reinsurance strategies of all members tend to converge under the JP-optimal and BO1 contracts as $\gamma_R$ becomes sufficiently large. This occurs because the reinsurer's risk-bearing capacity decreases as she becomes more risk averse, reducing the advantage of the JP-optimal contracts over the BO1 contracts. Consequently, the distinction between the two designs becomes less significant, leading to the alignment of the reinsurance strategies.
  }


\subsubsection{Total welfare gains}


\begin{figure}[!ht]
    \centering
    \includegraphics[scale=0.325]{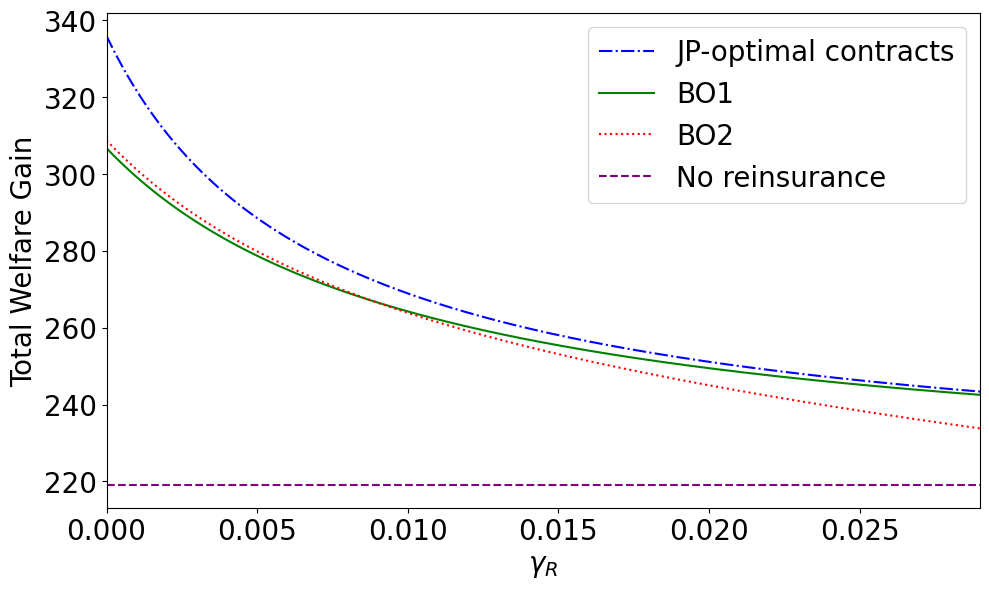}
    \caption{Comparison of total welfare gains among contracts with different $\gamma_R$.}
    \label{fig:TWGvsGammaR}
\end{figure}

Fig.~\ref{fig:TWGvsGammaR}  shows the total welfare gains of all contracts as $\gamma_R$ varies, where JP-optimal contracts attain the largest total welfare levels and thus corroborate the result in Corollary \ref{Cor:WelfareGain}.
{We highlight three key observations below.}

First, the welfare levels of both Pareto and Bowley contracts decline with $\gamma_R$. 
{As reflected in the decline of members’ reinsurance strategies with $\gamma_R$ depicted in Fig.~\ref{fig:pVSGammaR}, the effectiveness of the risk-transfer mechanism weakens as the risk-bearing capacity of the reinsurer decreases.}
Such dampening of risk sharing 
leads to a decline in the welfare improvement.

{Second, concerning the impact of the single-loading restriction, we observe that when $\gamma_R$ is small, the total welfare gain under the BO2 contract mildly exceeds that under the BO1 contract. As the reinsurer has a higher risk-bearing capacity, Member 2, who bears the riskiest loss, can take advantage of the loading design in the BO2 contract to transfer her risk to the reinsurer. This disproportionately increases her welfare gain, which dominates the welfare improvement of the community, thereby leading to a lower total welfare gain (cf. Table \ref{Tab:WelfareGains}). As $\gamma_R$ increases, the risk-bearing capacity is weakened, which calls for a more balanced distribution of risks in the community. In that case, benefiting from the better risk-premium alignment, the BO1 contract outperforms the BO2 contract and delivers a higher total welfare gain. }

{Third, the total welfare gains of the JP-optimal and the BO1 contracts converge as $\gamma_R$ increases. This is consistent with the findings in Fig.~\ref{fig:pVSGammaR}, which reveal the diminishing advantage of the Pareto design as the reinsurer’s risk-bearing capacity declines, ultimately aligning both the reinsurance strategies and the total welfare outcomes of the two designs. }

\section{Conclusion}\label{Sec:Conclusion}
This paper introduces a new design of P2P insurance schemes with an endogenous 
reinsurance treaty.
We compare two canonical designs depending on the interaction between the plan manager and the reinsurer: Pareto and Bowley.

In the Pareto design, full cooperation induces pricing freedom, corresponding to the allocation of welfare gains among agents. By formulating a transferable-utility cooperative game, we show that this allocation problem admits a non-empty core, and any core allocation induces a coalitionally stable JP-optimal contract. 
In contrast, the Bowley design leads to strategic pricing distortions. 
We show that it is never JP-optimal and yields strictly lower total welfare than its cooperative counterpart. This arises from the reinsurer’s first-mover advantage, which leads to higher reinsurance prices and alters reinsurance demand and mutualization patterns. 

Our numerical analysis showcases the welfare improvement of the reinsurance option and illustrates how single-loading restrictions further affect welfare distribution, particularly disadvantaging low-risk members. 
Also, the comparative static analysis highlights 
that a less risk-averse reinsurer accepts more risk from members, leading to a greater increase in total welfare.

Having opened up a new avenue of game-theoretic study on reinsurance contracting in P2P insurance, our framework can be generalized to other settings: 
(i) A continuous-time risk-exchange environment \citep{tao2025pareto}, which can also incorporate correlated investment risk, thereby enabling designs with collective investment \citep{balter2024robust,ng2025pareto}; 
(ii) Multiple reinsurers: this enables different strategic interaction between reinsurers such as cooperation, co-opetition and competition; see, e.g., \cite{cao2025co} and \cite{chu2025mean};
(iii) The Nash bargaining model \citep{boonen2016nash,nguyen2025optimal}: it is of both theoretical and practical interest to study differences in bargaining power between the reinsurer and the P2P plan manager. 
We leave these for future research.



\section*{Acknowledgments}		
		Tak Wa Ng and Thai Nguyen acknowledge the support from the Natural Sciences and Engineering Research Council of Canada (Grant No. RGPIN-2021-02594) and the Chair of Actuary, Laval University. Tak Wa Ng also acknowledges the hospitality of the Department of Mathematics at The Ohio State University during the drafting of this paper. Kenneth Ng acknowledges support from the start-up fund at The Ohio State University and the CKER research fund of the Society of Actuaries (Project title: Pricing and Staking of Decentralized Insurance).

\begingroup
\setlength{\bibsep}{0.25\baselineskip} 
\bibliographystyle{apalike}
\bibliography{sample}
\endgroup


\appendix
\renewcommand{\thesubsection}{\Alph{subsection}}

\section{Proof of Lemma \ref{Lem:Mbar}}\label{App:Mbar}


To show that $\overline{\bm{M}}$ is positive definite, consider for any $\bm{v}\in \mathbb{R}^n$, 
    \begin{align*}
        \bm{v}^\top \overline{\bm{M}}\bm{v} &= \gamma_R\bm{v}^\top\bm{\Sigma}\bm{v} + k\bm{v}^\top\bm{D}(\bm{\mu}^2)\bm{D}(\bm{\gamma})\bm{v} +  \frac{1}{\sum_{j=1}^n\gamma_j^{-1}}\left( \bm{v}^\top \bm{\Sigma}\bm{v} - k(\bm{\mu}^\top\bm{v})^2  \right).   
    \end{align*}
Using the Cauchy--Schwarz inequality and the positive-definiteness of $\bm{\Sigma}$,  
\[
(\bm{\mu}^\top \bm{v})^2 = \big((\bm{\Sigma}^{1/2}\bm{v})^\top (\bm{\Sigma}^{-1/2}\bm{\mu})\big)^2 \le (\bm{v}^\top \bm{\Sigma} \bm{v}) (\bm{\mu}^\top \bm{\Sigma}^{-1} \bm{\mu}) = \frac{1}{k}(\bm{v}^\top \bm{\Sigma} \bm{v}). 
\]
Therefore, for $\bm{v}\neq 0$, we have $\bm{v}^\top \overline{\bm{M}}\bm{v} \geq  \gamma_R\bm{v}^\top\bm{\Sigma}\bm{v} + k\bm{v}^\top\bm{D}(\bm{\mu}^2)\bm{D}(\bm{\gamma})\bm{v}  >0$.

\section{Proof of Proposition \ref{Prop:ConvexPareto}}\label{App:ConvexPareto}

Due to the actuarial fairness condition, the objective can be transformed to $\frac{1}{2}tr(\bm{D}(\bm{\gamma})\bm{A}\bm{\Sigma}\bm{A}^\top)+\frac{\gamma_{R}}{2}\bm{p}^\top\bm{\Sigma}\bm{p}.$
Consider the corresponding Lagrangian:
{\small\begin{align*}
\mathcal{L}(\bm{A},\bm{p},\bm{\phi},\bm{\psi})
=\frac{1}{2}tr({\bm{D}(\bm{\gamma})}\bm{A}\bm{\Sigma}\bm{A}^\top)+\frac{\gamma_{R}}{2}\bm{p}^\top\bm{\Sigma}\bm{p}+(\bm{1}^\top-\bm{1}^\top\bm{A}-\bm{p}^\top)\bm{\phi}+\bm{\psi}^\top (\bm{\mu}-\bm{D}(\bm{\mu})\bm{p}-\bm{A}\bm{\mu}),
\end{align*}}%
where $\bm{\phi}$ and $\bm{\psi}$ are Lagrangian multipliers capturing two constraints in Problem \eqref{Prob:RS}.
The first order condition (FOC) then yields
{\small\begin{align}\label{Eq:FOCL2}
\begin{split}
&\frac{\partial \mathcal{L}}{\partial \bm{A}}=\bm{D}(\bm{\gamma})\bm{A}\bm{\Sigma}-\bm{1}\bm{\phi}^\top-\bm{\psi}\bm{\mu}^\top=\bm{0}_{n\times n},\quad
\frac{\partial \mathcal{L}}{\partial \bm{p}}=\gamma_{R}\bm{\Sigma}\bm{p}-\bm{D}(\bm{\mu})\bm{\psi}-\bm{\phi}=\bm{0},\\
&\frac{\partial \mathcal{L}}{\partial \bm{\phi}}=\bm{1}^\top-\bm{1}^\top\bm{A}-\bm{p}^\top=\bm{0},\quad
\frac{\partial \mathcal{L}}{\partial \bm{\psi}}=\bm{\mu}-\bm{D}(\bm{\mu})\bm{p}-\bm{A}\bm{\mu}=\bm{0},
\end{split}
\end{align}}%
where the first equation follows from \citet[Equation (114)]{petersen2008matrix} and the cyclic property of the trace operator.

From the first equation in \eqref{Eq:FOCL2}, we have $ \bm{A}=\bm{D}(\bm{\gamma})^{-1}(\bm{1}\bm{\phi}^\top+\bm{\psi}\bm{\mu}^\top)\bm{\Sigma}^{-1}$.
Plugging this into $\bm{1}^\top\bm{A}=\bm{1}^\top-\bm{p}^\top$, we have $\bm{1}^\top\bm{D}(\bm{\gamma})^{-1}(\bm{1}\bm{\phi}^\top+\bm{\psi}\bm{\mu}^\top)\bm{\Sigma}^{-1}=\bm{1}^\top-\bm{p}^\top$,
which leads to 
   {\small  \begin{equation*}
        \bm{\phi}^\top=\frac{\bm{1}^\top}{\sum_{j=1}^n\gamma_j^{-1}}(\bm{D}(\bm{1}-\bm{p})\bm{\Sigma}-\bm{D}(\bm{\gamma})^{-1}\bm{\psi}\bm{\mu}^\top).
    \end{equation*}}

Substituting $ \bm{A}=\bm{D}(\bm{\gamma})^{-1}(\bm{1}\bm{\phi}^\top+\bm{\psi}\bm{\mu}^\top)\bm{\Sigma}^{-1}$ and $\bm{\phi}^\top$ to $\bm{A}\bm{\mu}=\bm{D}(\bm{1}-\bm{p})\bm{\mu}$, we obtain
 {\small \begin{align*}
      &\bm{D}(\bm{1}-\bm{p})\bm{\mu}
       = \bm{D}(\bm{\gamma})^{-1}\left(\frac{\bm{1}\bm{1}^\top}{\sum_{j=1}^n\gamma_j^{-1}}(\bm{D}(\bm{1}-\bm{p})\bm{\Sigma}-\bm{D}(\bm{\gamma})^{-1}\bm{\psi}\bm{\mu}^\top)+\bm{\psi}\bm{\mu}^\top\right)\bm{\Sigma}^{-1}\bm{\mu}\\
\iff&\left(\bm{D}(\bm{\gamma})-\frac{\bm{1}\bm{1}^\top}{\sum_{j=1}^n\gamma_j^{-1}}\right)\bm{D}(\bm{1}-\bm{p})\bm{\mu}=k^{-1}\left(\bm{I}_n-\frac{\bm{1}\bm{1}^\top\bm{D}(\bm{\gamma})^{-1}}{\sum_{j=1}^n\gamma_j^{-1}}\right)\bm{\psi}.
    \end{align*}}%
Solving the equation yields 
 {\small    \begin{equation}
    \label{eq:psi:phi:pareto}
        \bm{\psi}=k\bm{D}(\bm{\gamma})\bm{D}(\bm{1}-\bm{p})\bm{\mu} \quad \text{and} \quad \bm{\phi}^\top=\frac{\bm{1}^\top}{\sum_{j=1}^n\gamma_j^{-1}}(\bm{D}(\bm{1}-\bm{p})\bm{\Sigma}-k\bm{D}(\bm{1}-\bm{p})\bm{\mu}\bm{\mu}^\top).
    \end{equation}}%
Substituting \eqref{eq:psi:phi:pareto} into $\bm{A}=\bm{D}(\bm{\gamma})^{-1}(\bm{1}\bm{\phi}^\top+\bm{\psi}\bm{\mu}^\top)\bm{\Sigma}^{-1}$, {we deduce that $\bm{A}_*$ is given by \eqref{Eq:OptAP2}. }  

To obtain $\bm{p}_*$, we substitute \eqref{eq:psi:phi:pareto} into the second equation of \eqref{Eq:FOCL2}, which leads to
\begin{align}\label{Eq:linp2}
\begin{split}    &\gamma_{R}\bm{\Sigma}\bm{p}=k\bm{D}(\bm{\mu})\bm{D}(\bm{\gamma})\bm{D}(\bm{1}-\bm{p})\bm{\mu}+\frac{1}{\sum_{j=1}^n\gamma_j^{-1}}(\bm{\Sigma}(\bm{1}-\bm{p})-k\bm{\mu}\bm{\mu}^\top(\bm{1}-\bm{p})).
\end{split}
\end{align}
Solving the equation yields \eqref{Eq:OptAP2}. 

Finally, we prove that the unique solution of the linear system \eqref{Eq:linp2} lies in $(0,1)^n$ given the condition \eqref{Eq:unicond2}. Note that the system can be re-expressed as: {\small 
\begin{equation}\label{Eq:pi2}
L_i(\bm{p}):=\sum_{m=1}^n\sigma_{im}\left(\frac{1-p_m}{\sum_{j=1}^n\gamma_j^{-1}}-\gamma_{R}p_m\right)+k\mu_i^2\gamma_i(1-p_i)-\frac{k\mu_i}{\sum_{j=1}^n\gamma_j^{-1}}\sum_{m=1}^n\mu_m(1-p_m)=0,
\end{equation}}%
for all $i=1,\ldots,n$. The claim can be established using the Poincaré–Miranda theorem. To this end, we need to check that for $\bm{p}\in[0,1]^n$, the sign of left-hand side in \eqref{Eq:pi2} is nonnegative (resp.~nonpositive) if $p_i=1$ (resp. $p_i=0$) for all $i=1,\ldots,n$. For $p_i=1$, we have  
   {\small   \begin{align*}
    &\ L_i(p_1,\dots,p_{i-1},1,p_{i+1},\dots,p_n) \\ 
     =&\ -\gamma_R \sigma_i^2 + \frac{\sum_{m\neq i} \sigma_{im}}{\sum_{j=1}^n\gamma_j^{-1}} - \frac{k\mu_i}{\sum_{j=1}^n\gamma_j^{-1}}\sum_{m\neq i}\mu_m + \sum_{m\neq i}p_m\left(\frac{k\mu_i\mu_m}{\sum_{j=1}^n\gamma_j^{-1}} - \left(\gamma_R+\frac{1}{\sum_{j=1}^n\gamma_j^{-1}}\right)\sigma_{im} \right)\\
     \leq&\ -\gamma_R \sigma_i^2 + \frac{\sum_{m\neq i} \sigma_{im}}{\sum_{j=1}^n\gamma_j^{-1}} - \frac{k\mu_i\sum_{m\neq i}\mu_m}{\sum_{j=1}^n\gamma_j^{-1}} + \sum_{m\neq i}\left(\frac{k\mu_i\mu_m}{\sum_{j=1}^n\gamma_j^{-1}}-\left(\gamma_R+\frac{1}{\sum_{j=1}^n\gamma_j^{-1}}\right)\sigma_{im} \right)_+  
     < 0. 
 \end{align*}}%
On the other hand, for $p_i=0$, we obtain 
{\small\begin{align*}
     &\ L_i(p_1,\dots,p_{i-1},0,p_{i+1},\dots,p_n) \\
    =&\  \frac{\sum_{m=1}^n\sigma_{im}}{\sum_{j=1}^n\gamma_j^{-1}} + k\mu_i^2\gamma_i - \frac{k\mu_i\sum_{m=1}^n\mu_m}{\sum_{j=1}^n\gamma_j^{-1}} +  \sum_{m\neq i} p_m\left( \frac{k\mu_i\mu_m}{\sum_{j=1}^n\gamma_j^{-1}} - \left(\gamma_R + \frac{1}{\sum_{j=1}^n\gamma_j^{-1}} \right)\sigma_{im}  \right) \\
    \geq&\  \frac{\sum_{m=1}^n\sigma_{im}}{\sum_{j=1}^n\gamma_j^{-1}} + k\mu_i^2\gamma_i - \frac{k\mu_i\sum_{m=1}^n\mu_m}{\sum_{j=1}^n\gamma_j^{-1}} - \sum_{m\neq i} \left( \left(\gamma_R + \frac{1}{\sum_{j=1}^n\gamma_j^{-1}} \right)\sigma_{im}-\frac{k\mu_i\mu_m}{\sum_{j=1}^n\gamma_j^{-1}}  \right)_+  
    > 0.
    \end{align*}}%

Note that the above two inequalities follow from the condition \eqref{Eq:unicond2}. In addition, as the inequalities in \eqref{Eq:unicond2} are strict, $p_i\neq0$ (resp. $p_i\neq1$) for all $i=1,\dots,n$. Finally, the uniqueness  follows from the invertibility of $\overline{\bm{M}}$ as asserted in Lemma \ref{Lem:Mbar}.

\section{Proof of Theorem \ref{Thm:JPO}}\label{App:JPO}

We first prove $\mathcal{Z}\subseteq\mathcal{JPO}$ by contradiction. Assume that there exists a contract $(\bm{\hat{A}},\bm{\hat{p}},\bm{\hat{\eta}})\in\mathcal{Z}$ that is not JP-optimal, then there exists $(\bm{\Tilde{A}},\bm{\Tilde{p}},\bm{\Tilde{\eta}})\in\mathcal{IR}$ such that $v(\bm{\Tilde{\eta}},\bm{\Tilde{p}})\leq v(\bm{\hat{\eta}},\bm{\hat{p}})$ and  $u_i(\bm{\Tilde{A}},\bm{\Tilde{p}},\bm{\Tilde{\eta}})\leq u_i(\bm{\hat{A}},\bm{\hat{p}},\bm{\hat{\eta}})$ 
for all $i=1,\dots,n$, with at least one strict inequality. This implies
$u(\bm{\Tilde{A}},\bm{\Tilde{p}},\bm{\Tilde{\eta}})+v(\bm{\Tilde{\eta}},\bm{\Tilde{p}})<u(\bm{\hat{A}},\bm{\hat{p}},\bm{\hat{\eta}})+v(\bm{\hat{\eta}},\bm{\hat{p}})
$, which violates the optimality of $(\bm{\hat{A}},\bm{\hat{p}},\bm{\hat{\eta}})$.

Next, we prove that $\mathcal{JPO}\subseteq\mathcal{Z}$, again by contradiction. Assume that there exists $(\bm{\underline{A}},\bm{\underline{p}},\bm{\underline{\eta}})\in\mathcal{JPO}$ that does not belong to $\mathcal{Z}$. Then, there exists $(\bm{\check{A}},\bm{\check{p}},\bm{\check{\eta}})\in\mathcal{Z}$ such that $u(\bm{\check{A}},\bm{\check{p}},\bm{\check{\eta}})+v(\bm{\check{\eta}},\bm{\check{p}})<u(\bm{\underline{A}},\bm{\underline{p}},\bm{\underline{\eta}})+v(\bm{\underline{\eta}},\bm{\underline{p}})$. 
Define $\eta_i^\star$, $i=1,\dots,n$, by the solution of
\begin{equation}\label{Eq:etastar}    (1+\eta_i^\star)\check{p}_i\mu_i=(1+\check{\eta}_i)\check{p}_i\mu_i+u_i(\bm{\underline{A}},\bm{\underline{p}},\bm{\underline{\eta}})-u_i(\bm{\check{A}},\bm{\check{p}},\bm{\check{\eta}})=(1+\underline{\eta}_i)\underline{p}_i\mu_i+\rho_i(\bm{\underline{A}})-\rho_i(\bm{\check{A}}).
\end{equation}
Note that this needs $\check{p}_i>0$ for all $i=1,\dots,n$, which is ensured by Condition \eqref{Eq:unicond2}.
We then have
\begin{align*}   u_i(\bm{\check{A}},\bm{\check{p}},\bm{\eta}^\star)=&\rho_i(\bm{\check{A}})+(1+\eta_i^\star)\check{p}_i\mu_i
=(1+\underline{\eta}_i)\underline{p}_i\mu_i+\rho_i(\bm{\underline{A}})
=u_i(\bm{\underline{A}},\bm{\underline{p}},\bm{\underline{\eta}})
\end{align*}
for all $i=1,\dots,n$, where the second last equality follows from \eqref{Eq:etastar}.

On the other hand, summing \eqref{Eq:etastar} over $i=1,\dots,n$ yields $(\bm{D}(\bm{\mu})\bm{\eta}^\star)^\top\bm{\check{p}}=(\bm{D}(\bm{\mu})\bm{\check{\eta}})^\top\bm{\check{p}}+u(\bm{\underline{A}},\bm{\underline{p}},\bm{\underline{\eta}})-u(\bm{\check{A}},\bm{\check{p}},\bm{\check{\eta}})$. 
Using this and $u(\bm{\check{A}},\bm{\check{p}},\bm{\check{\eta}})+v(\bm{\check{\eta}},\bm{\check{p}})<u(\bm{\underline{A}},\bm{\underline{p}},\bm{\underline{\eta}})+v(\bm{\underline{\eta}},\bm{\underline{p}})$, 
\begin{align*}
    v(\bm{\eta^\star},\bm{\check{p}})=&\frac{\gamma_{R}}{2}\bm{\check{p}}^\top\bm{\Sigma}\bm{\check{p}}-(\bm{D}(\bm{\mu})\bm{\eta}^\star)^\top\bm{\check{p}}=v(\bm{\check{\eta}},\bm{\check{p}})+u(\bm{\check{A}},\bm{\check{p}},\bm{\check{\eta}})-u(\bm{\underline{A}},\bm{\underline{p}},\bm{\underline{\eta}})<v(\bm{\underline{\eta}},\bm{\underline{p}}),
\end{align*}
These imply $(\bm{\check{A}},\bm{\check{p}},\bm{\eta}^\star)\in\mathcal{IR}$ and Pareto-dominates $(\bm{\underline{A}},\bm{\underline{p}},\bm{\underline{\eta}})$, leading to the desired contradiction.

Next, we show that $\mathcal{Z}$ is non-empty. {To this end, for $(\bm{A}_*,\bm{p}_*)$ solving \eqref{Prob:RS}, we construct $\bm{\eta}_*$ such that $(\bm{A}_*,\bm{p}_*,\bm{\eta}_*)\in \mathcal{IR}$.} For any $\bm{\eta}$, let $\varepsilon:= u(\bm{I},0,\bm{\eta})+v(\bm{\eta},\bm{0}) -u(\bm{A}_*,\bm{p}_*,\bm{\eta}) - v(\bm{\eta},\bm{p}_*)$. Note that $\varepsilon>0$ and is independent of $\bm{\eta}$, since $(\bm{A}_*,\bm{p}_*)$ is the unique minimizer of \eqref{Prob:RS}. Define $\bm{\eta}_*=(\eta_{*1},\dots,\eta_{*n})$ by,
       \begin{equation}
        \label{Eq:eta*:IR}
            \eta_{*i}p_{i*}\mu_i = \frac{\gamma_i}{2}\left( \sigma_i^2 - \bm{A}_{*i}\bm{\Sigma}\bm{A}_{*i}^\top \right) - \frac{\varepsilon}{2n}, \quad i=1,\dots,n,
         \end{equation}   
    It is clear that the members' IR constraints \eqref{Eq:MemberIR} are satisfied by the definition of $\eta_{*i}$ with welfare gain 
        \begin{equation*}
            \omega_i(\bm{A}_*,\bm{p}_*,\bm{\eta}_*) = \mu_i + \frac{\gamma_i\sigma_i^2}{2}- u_i(\bm{A}_*,\bm{p}_*,\bm{\eta}_*) = \frac{\gamma_i}{2}\left( \sigma_i^2 -  \bm{A}_{*i}\bm{\Sigma}\bm{A}_{*i}^\top  \right) -  \eta_{*i}p_{i*}\mu_i = \frac{\varepsilon}{2n}>0. 
        \end{equation*}

    Summing \eqref{Eq:eta*:IR} over $i=1,\dots,n$, we have 
        \begin{equation*}
            (\bm{D}(\bm{\mu})\bm{\eta}_*)^\top\bm{p}_* = \frac{1}{2}\left( tr(\bm{D}(\bm{\gamma})\bm{\Sigma}) - tr(\bm{D}(\bm{\gamma})\bm{A}_*\bm{\Sigma}\bm{A}_*^\top) \right)  - \frac{\varepsilon}{2}. 
        \end{equation*}
    Using this and the actuarial fairness condition, we have {\small 
        \begin{align*}
         \omega_R(\bm{A}_*,\bm{p}_*,\bm{\eta}_*)&=       (\bm{D}(\bm{\mu})\bm{\eta}_*)^\top\bm{p}_* - \frac{\gamma_R}{2}\bm{p}_*^\top\bm{\Sigma}\bm{p}_* \\
            &=  \frac{1}{2}\left( tr(\bm{D}(\bm{\gamma})\bm{\Sigma}) - tr(\bm{D}(\bm{\gamma})\bm{A}_*\bm{\Sigma}\bm{A}_*^\top) \right) -  \frac{\gamma_R}{2}\bm{p}_*^\top\bm{\Sigma}\bm{p}_* - \frac{\varepsilon}{2} \\
            &= \frac{1}{2}tr(\bm{D}(\bm{\gamma})\bm{\Sigma}) + \bm{1}^\top \bm{A}_*\bm{\mu} + \bm{\mu}^\top\bm{p}_* \\
            &\quad -\left( \frac{\gamma_R}{2}\bm{p}_*^\top\bm{\Sigma}\bm{p}_* + \frac{1}{2}tr(\bm{D}(\bm{\gamma})\bm{A}_*\bm{\Sigma}\bm{A}_*^\top) + \bm{1}^\top \bm{A}_*\bm{\mu} + \bm{\mu}^\top\bm{p}_* \right) - \frac{\varepsilon}{2} \\
            &=  \frac{1}{2}tr(\bm{D}(\bm{\gamma})\bm{\Sigma})  + \bm{1}^\top\bm{\mu} -\left(u(\bm{A}_*,\bm{p}_*,\bm{\eta}_*) + v(\bm{\eta}_*,\bm{p}_*) \right) - \frac{\varepsilon}{2} \\
            &= u(\bm{I},\bm{0},\bm{\eta}_*) + v(\bm{\eta}_*,\bm{0})  - \left(u(\bm{A}_*,\bm{p}_*,\bm{\eta}_*) + v(\bm{\eta}_*,\bm{p}_*)   \right) - \frac{\varepsilon}{2}   =\frac{\varepsilon}{2} >0,
        \end{align*}}%
     Therefore, the reinsurer's IR constraint is also satisfied.

To study the equivalence between (i) and (ii), we define the following auxiliary problem:
\begin{equation}
\label{Prob:auxJPO}
    \min_{(\bm{A},\bm{p},\bm{\eta})\in\mathcal{IR}}v(\bm{\eta},\bm{p})\quad\text{s.t. }u_i(\bm{A},\bm{p},\bm{\eta})\leq c_i\text{ for all }i=1,\dots,n,\  \bm{A}\bm{\mu}+\bm{D}(\bm{\mu})\bm{p}=\bm{\mu},\ \mathbf{1}^\top\bm{A}+\bm{p}^\top=\bm{1}^\top.
\end{equation}
We show that $(\bm{A}_*,\bm{p}_*,\bm{\eta}_*)\in\mathcal{JPO}$ if and only if there exists $(c_i)_{i=1}^n$ such that $(\bm{A}_*,\bm{p}_*,\bm{\eta}_*)$ solves \eqref{Prob:auxJPO}.

Suppose that $(\bm{A}_*,\bm{p}_*,\bm{\eta}_*)\in\mathcal{JPO}$ and let $c_i=u_i(\bm{A}_*,\bm{p}_*,\bm{\eta}_*)$, $i=1,\dots,n$. We need to show that $(\bm{A}_*,\bm{p}_*,\bm{\eta}_*)$ is the minimizer of Problem \eqref{Prob:auxJPO} with parameters $c_1,\dots,c_n$.
Assume the contrary, so that there exists $(\bm{A}',\bm{p}',\bm{\eta}')\in\mathcal{IR}$ with $u_i(\bm{A}',\bm{p}',\bm{\eta}')\leq c_i= u_i(\bm{A}_*,\bm{p}_*,\bm{\eta}_*)$ for all $i=1,\dots,n$, and $v(\bm{\eta}',\bm{p}')< v(\bm{\eta}_*,\bm{p}_*)$. 
This contradicts with the JP optimality of $(\bm{A}_*,\bm{p}_*,\bm{\eta}_*)$. 

Conversely, assume that there exists $c_1,\dots,c_n$ such that $(\bm{A}_*,\bm{p}_*,\bm{\eta}_*)\notin\mathcal{JPO}$ is optimal for Problem \eqref{Prob:auxJPO}. Then, there exists $(\bm{A}_{'},\bm{p}_{'},\bm{\eta}_{'})\in\mathcal{IR}$ such that $u_i(\bm{A}_{'},\bm{p}_{'},\bm{\eta}_{'})\leq c_i= u_i(\bm{A}_*,\bm{p}_*,\bm{\eta}_*)$ for all $i=1,\dots,n$, and $v(\bm{\eta}_{'},\bm{p}_{'})\leq v(\bm{\eta}_*,\bm{p}_*)$, 
with at least one strict inequality. If $v(\bm{\eta}_{'},\bm{p}_{'})< v(\bm{\eta}_*,\bm{p}_*)$, then it contradicts the optimality of $(\bm{A}_*,\bm{p}_*,\bm{\eta}_*)$. Thus, we have $ v(\bm{\eta}_{'},\bm{p}_{'})= v(\bm{\eta}_*,\bm{p}_*)$,
and there exists $j\in\{1,\dots,n\}$ such that $  u_j(\bm{A}_{'},\bm{p}_{'},\bm{\eta}_{'})<u_j(\bm{A}_*,\bm{p}_*,\bm{\eta}_*)\leq c_j$.
Define $\epsilon_j:=u_j(\bm{A}_*,\bm{p}_*,\bm{\eta}_*)-u_j(\bm{A}_{'},\bm{p}_{'},\bm{\eta}_{'})$ and consider a new contract $(\bm{A}_{'},\bm{p}_{'},\bm{\eta}_>)$, where 
$\bm{\eta}_>=(\eta_{>1},\dots,\eta_{>n})$ is given by $\eta_{>i} := \eta_{'i}$ for $i\neq j$, and $\eta_{>j} := \eta_{'j} + \varepsilon_j/(p_{'j}\mu_j)$. Then, $u_i(\bm{A}_{'},\bm{p}_{'},\bm{\eta}_{>}) = u_i(\bm{A}_{'},\bm{p}_{'},\bm{\eta}_{'})\leq c_i$ for $i\neq j$, and $u_j(\bm{A}_{'},\bm{p}_{'},\bm{\eta}_{>}) = u_j(\bm{A}_{'},\bm{p}_{'},\bm{\eta}_{'}) + \varepsilon_j = u_j(\bm{A}_*,\bm{p}_*,\bm{\eta}_*) \leq c_j$. 
In addition, by the fact that $v(\bm{\eta}_{'},\bm{p}_{'})= v(\bm{\eta}_*,\bm{p}_*)$, we have $v(\bm{\eta}_>,\bm{p}_{'})=v(\bm{\eta}_{'},\bm{p}_{'})-\epsilon_j=v(\bm{\eta}_*,\bm{p}_*)-\epsilon_j<v(\bm{\eta}_*,\bm{p}_*)$,
contradicting the optimality of $(\bm{A}_*,\bm{p}_*,\bm{\eta}_*)$.

Finally, we need to show the equivalence of Problems \eqref{Prob:auxJPO} and \eqref{Prob:equJPO}, which follows from the fact that optimality of Problem \eqref{Prob:auxJPO} only occurs when constraints are binding; otherwise, one can repeat the above proof with a new contract that surcharges the member without a binding constraint to improve the objective value.  

To show the equivalence between (i) and (iii), we introduce the following auxiliary problem:
\begin{equation}\label{Prob:auxJPO2}
    \min_{(\bm{A},\bm{p},\bm{\eta})\in\mathcal{IR},{\bm{p}>0}}u(\bm{A},\bm{p},\bm{\eta})\quad\text{s.t. }v(\bm{\eta},\bm{p})\leq c_R,\ \bm{A}\bm{\mu}+\bm{D}(\bm{\mu})\bm{p}=\bm{\mu},\ \mathbf{1}^\top\bm{A}+\bm{p}^\top=\bm{1}^\top.
\end{equation}
Suppose that $(\bm{A}_*,\bm{p}_*,\bm{\eta}_*)\in\mathcal{JPO}$ and let $c_R=v(\bm{\eta}_*,\bm{p}_*)$. Assume the contrary that $(\bm{A}_*,\bm{p}_*,\bm{\eta}_*)$ does not solve Problem \eqref{Prob:auxJPO2}. Then, there exists $(\bm{A}^{"},\bm{p}^{"},\bm{\eta}^{"})\in\mathcal{IR}$ that solves Problem \eqref{Prob:auxJPO2}, with $u(\bm{A}^{"},\bm{p}^{"},\bm{\eta}^{"})<u(\bm{A}_*,\bm{p}_*,\bm{\eta}_*)$ and $v(\bm{\eta}^{"},\bm{p}^{"})\leq c_R$.
Fix $i\in\{1,\ldots,n\}$, and for $j\neq i$, define $\overline{\epsilon}_j:=u_j(\bm{A}_*,\bm{p}_*,\bm{\eta}_*)-u_j(\bm{A}^{"},\bm{p}^{"},\bm{\eta}^{"})$, $\eta_j^>=\eta_j^{"}+\frac{\overline{\epsilon}_j}{p_j^{"}\mu_j}$, and $\eta_i^>=\eta_i^{"}-\frac{\sum_{j\neq i}\overline{\epsilon}_j}{p_i^{"}\mu_i}$. 
 Then, we have $u_j(\bm{A}^{"},\bm{p}^{"},\bm{\eta}^>)=u_j(\bm{A}^{"},\bm{p}^{"},\bm{\eta}^{"})+\overline{\epsilon}_j=u_j(\bm{A}_*,\bm{p}_*,\bm{\eta}_*)$, $v(\bm{\eta}^>,\bm{p}^{"})=v(\bm{\eta}^{"},\bm{p}^{"})$, and
{\small\begin{align*}
u_i(\bm{A}^{"},\bm{p}^{"},\bm{\eta}^>)=&u_i(\bm{A}^{"},\bm{p}^{"},\bm{\eta}^{"})-\sum_{j\neq i}{\bar{\epsilon}_j}\\
=&u_i(\bm{A}^{"},\bm{p}^{"},\bm{\eta}^{"})-\sum_{j\neq i}(u_j(\bm{A}_*,\bm{p}_*,\bm{\eta}_*)-u_j(\bm{A}^{"},\bm{p}^{"},\bm{\eta}^{"}))\\
=&u(\bm{A}^{"},\bm{p}^{"},\bm{\eta}^{"})-\sum_{j\neq i}u_j(\bm{A}_*,\bm{p}_*,\bm{\eta}_*)\\
=&u(\bm{A}^{"},\bm{p}^{"},\bm{\eta}^{"})-u(\bm{A}_*,\bm{p}_*,\bm{\eta}_*)+u_i(\bm{A}_*,\bm{p}_*,\bm{\eta}_*)<u_i(\bm{A}_*,\bm{p}_*,\bm{\eta}_*),
\end{align*}}%
contradicting the JP-optimality of $(\bm{A}_*,\bm{p}_*,\bm{\eta}_*)$.

Conversely, suppose that $(\bm{A}_*,\bm{p}_*,\bm{\eta}_*)$ solves Problem \eqref{Prob:auxJPO2}. Assume the contrary that $(\bm{A}_*,\bm{p}_*,\bm{\eta}_*)$ is not JP-optimal and there exists $(\bm{A}^@,\bm{p}^@,\bm{\eta}^@)\in\mathcal{IR}$ such that $u_i(\bm{A}^@,\bm{p}^@,\bm{\eta}^@)\leq u_i(\bm{A}_*,\bm{p}_*,\bm{\eta}_*)$ for all $i=1,\ldots,n$, and $v(\bm{\eta}^@,\bm{p}^@)\leq v(\bm{\eta}_*,\bm{p}_*)$, 
with at least one of these inequalities being strict. The first $n$ inequalities must be equalities; otherwise, it violates the optimality of $(\bm{A}_*,\bm{p}_*,\bm{\eta}_*)$. Then, we have $v(\bm{\eta}^@,\bm{p}^@)<v(\bm{\eta}_*,\bm{p}_*)$.
Define $\underline{\epsilon}=v(\bm{\eta}_*,\bm{p}_*)-v(\bm{\eta}^@,\bm{p}^@)>0$ and
$\eta_i^<=\eta_i^@-\underline{\epsilon}/np_i^@\mu_i$. Then, $u_i(\bm{A}^@,\bm{p}^@,\bm{\eta}^<)=u_i(\bm{A}^@,\bm{p}^@,\bm{\eta}^@)-\frac{\underline{\epsilon}}{n}<u_i(\bm{A}^@,\bm{p}^@,\bm{\eta}^@)$.
This implies $u(\bm{A}^@,\bm{p}^@,\bm{\eta}^<)<u(\bm{A}^@,\bm{p}^@,\bm{\eta}^@)=u(\bm{A}_*,\bm{p}_*,\bm{\eta}_*)$ and $v(\bm{\eta}_*,\bm{p}_*)=v(\bm{\eta}^@,\bm{p}^@)+\underline{\epsilon}=v(\bm{\eta}^<,\bm{p}^@)$, contradicting the optimality of $(\bm{A}_*,\bm{p}_*,\bm{\eta}_*)$.

Finally, we observe that Problem \eqref{Prob:auxJPO2} reaches optimality when the inequality constraint becomes an equality, i.e., Problem \eqref{Prob:equJPO2}; otherwise, the loadings can be adjusted to improve the objective value.

\section{Proof of Theorem \ref{Thm:CoreNB}}\label{App:CoreNB}

We first define an auxiliary game $(\mathcal{T},b)$ that includes only coalitions with the reinsurer, where $\mathcal{T}:=\{1,\dots,n\}$ and $b:2^{\mathcal{T}}\mapsto\mathbb{R}$ such that $b(\mathcal{S}\cup\mathcal{R}):=\mathcal{B}(\mathcal{S}\cup\mathcal{R})$.
The corresponding core is defined as 
{\small\begin{equation*}
    core(\mathcal{T},b):=\left\{\bm{c}\in\mathbb{R}^{n+1}_+:\sum_{i\in\mathcal{S}\cup\mathcal{R}}c_i\geq b(\mathcal{S}\cup\mathcal{R}),\ \sum_{i\in\mathcal{T}\cup\mathcal{R}}{c_i}=b(\mathcal{T}\cup\mathcal{R}),\  \text{ for all } \emptyset\neq\mathcal{S}\subseteq\mathcal{T}\right\}. 
\end{equation*}}%

Note that $core(\mathcal{T},b)$ is not empty since $(0,\dots,0,b(\mathcal{T}\cup\mathcal{R}))\in core(\mathcal{T},b)$. By the Bondareva-Shapley theorem \citep[see, e.g.,][Proposition 262.1]{osborne1994course}, $(\mathcal{T},b)$ is balanced, i.e., $\sum_{\mathcal{S}\subseteq\mathcal{T}}\lambda_{\mathcal{S}}b(\mathcal{S}\cup\mathcal{R})\leq b(\mathcal{T}\cup\mathcal{R})$
for all 
$\lambda_{\mathcal{S}}\in[0,1]$ such that $\sum_{\mathcal{S}\subseteq\mathcal{T}}\lambda_{\mathcal{S}}\bm{e}_{\mathcal{S}}=\bm{e}_{\mathcal{T}}$, where $\bm{e}_{\mathcal{S}}\in\mathbb{R}^n$ is the characteristic vector of the subset $\mathcal{S}$ such that $(\bm{e}_{\mathcal{S}})_i=1$ if $i\in\mathcal{S}$, and $(\bm{e}_{\mathcal{S}})_i=0$ otherwise. 

Next,
for any $\mathcal{C}\subsetneq  \mathcal{N}$ with  $R\notin\mathcal{C}$, let   $(\bm{A}_*^{\mathcal{C}\cup\mathcal{R}},\bm{p}_*^{\mathcal{C}\cup\mathcal{R}})$ and $\bm{A}_*^\mathcal{C}$ be the optimal solution to Problem \eqref{Prob:RS:C} and Problem \eqref{Prob:RS2}, respectively. 
Since Problem \eqref{Prob:RS2} can be considered as Problem \eqref{Prob:RS:C} with an additional constraint $\bm{p}^\mathcal{C}=\bm{0}^\mathcal{C}$, we have 
{\small\begin{equation}\label{Eq:OVcompare}    \sum_{i\in\mathcal{C}}\rho_i(\bm{A}_*^{\mathcal{C}\cup\mathcal{R}};\mathcal{C})+\rho_R(\bm{p}_*^{\mathcal{C}\cup\mathcal{R}};\mathcal{C}\cup\mathcal{R})\leq \sum_{i\in\mathcal{C}}\rho_i(\bm{A}_*^\mathcal{C};\mathcal{C}). 
\end{equation}}

To prove $core(\mathcal{N},\mathcal{B})\neq\emptyset$, we again invoke the Bondareva-Shapley theorem and show that $(\mathcal{N},\mathcal{B})$ is balanced. 
Let $\bm{\overline{e}}_{\mathcal{C}}\in\mathbb{R}^{n+1}$ be the characteristic vector of the subset $\mathcal{C}\subseteq\mathcal{N}$ such that $(\bm{\overline{e}}_{\mathcal{C}})_i=1$ if $i\in\mathcal{C}$, and $(\bm{\overline{e}}_{\mathcal{C}})_i=0$ otherwise. 
For any {$\mathcal{C}\subseteq \mathcal{N}$ and } $\overline{\lambda}_{\mathcal{C}}\in[0,1]$ such that $\sum_{\mathcal{C}\subseteq\mathcal{N}}\overline{\lambda}_{\mathcal{C}}\bm{\overline{e}}_{\mathcal{C}}=\bm{\overline{e}}_{\mathcal{N}}$, 
we have{\small 
\begin{align*}
    \sum_{\mathcal{C}\subseteq\mathcal{N}}\overline{\lambda}_{\mathcal{C}}\mathcal{B}(\mathcal{C})
    =&\sum_{\mathcal{C}\subseteq\mathcal{N}}\overline{\lambda}_{\mathcal{C}}\left[1_{R\in\mathcal{C}}\left(\sum_{i\in\mathcal{C}\backslash\mathcal{R}}\left(\mu_i+\frac{\gamma_i\sigma_i^2}{2}\right)- \sum_{i\in\mathcal{C}\backslash\mathcal{R}}\rho_i(\bm{A}_*^\mathcal{C};\mathcal{C}\backslash\mathcal{R})-\rho_R(\bm{p}_*^\mathcal{C};{\mathcal{C}}) \right)\right.\\
    &+\left.1_{R\notin\mathcal{C}}\left(\sum_{i\in\mathcal{C}}\left(\mu_i+\frac{\gamma_i\sigma_i^2}{2}\right)-\sum_{i\in\mathcal{C}}\rho_i(\bm{A}_*^{\mathcal{C}};\mathcal{C})\right)\right]\\
    \leq &\sum_{\mathcal{C}\subseteq\mathcal{N}}\overline{\lambda}_{\mathcal{C}}\left[1_{R\in\mathcal{C}}\left(\sum_{i\in\mathcal{C}\backslash\mathcal{R}}\left(\mu_i+\frac{\gamma_i\sigma_i^2}{2}\right)-\left(\sum_{i\in\mathcal{C}\backslash\mathcal{R}}\rho_i(\bm{A}_*^\mathcal{C};\mathcal{C}\backslash\mathcal{R})+\rho_R(\bm{p}_*^\mathcal{C};\mathcal{C})\right)\right)\right.\\
    &\left.+1_{R\notin\mathcal{C}}\left(\sum_{i\in\mathcal{C}}\left(\mu_i+\frac{\gamma_i\sigma_i^2}{2}\right)-\left(\sum_{i\in\mathcal{C}}\rho_i(\bm{A}_*^{\mathcal{C}\cup\mathcal{R}};\mathcal{C})+\rho_R(\bm{p}_*^\mathcal{C};\mathcal{C}\cup\mathcal{R})\right)\right)\right]\\
    =&\ \sum_{\mathcal{C}\subsetneq\mathcal{N},\mathcal{C}\ni R}\overline{\lambda}_{\mathcal{C}}b(\mathcal{C})+\sum_{\mathcal{C}\subsetneq\mathcal{N},\mathcal{C}\not\ni R}\overline{\lambda}_{\mathcal{C}}b(\mathcal{C}\cup\mathcal{R})\\
    =&\ {\sum_{\mathcal{S}\subseteq\mathcal{T}}\overline{\lambda}_{\mathcal{S}\cup\mathcal{R}} b(\mathcal{S}\cup\mathcal{R}) + \sum_{\mathcal{S}\subseteq\mathcal{T}}\overline{\lambda}_{\mathcal{S}} b(\mathcal{S}\cup\mathcal{R})  }
    \leq b(\mathcal{T}\cup\mathcal{R})=\mathcal{B}(\mathcal{N}),
\end{align*}}%
where the first inequality follows from \eqref{Eq:OVcompare}
, and
the second inequality follows from  the fact that $(\mathcal{T},b)$ is balanced and $\sum_{\mathcal{S}\subseteq\mathcal{T}}(\overline{\lambda}_{\mathcal{S}\cup\mathcal{R}}+\overline{\lambda}_{\mathcal{S}})\bm{e}_{\mathcal{S}}=\bm{e}_{\mathcal{T}}$. To see the last relation, note that  
{\small \begin{align*}
    \bm{\overline{e}}_N = \sum_{\mathcal{C}\subseteq\mathcal{N}}\overline{\lambda}_{\mathcal{C}}\bm{\overline{e}}_{\mathcal{C}}  = \sum_{\mathcal{S}\subseteq\mathcal{T}}\overline{\lambda}_{\mathcal{S}}\bm{\overline{e}}_{\mathcal{S}} + \sum_{\mathcal{S}\subseteq\mathcal{T}}\overline{\lambda}_{\mathcal{S}\cup\mathcal{R}}\bm{\overline{e}}_{\mathcal{S}\cup\mathcal{R}} &=\sum_{\mathcal{S}\subseteq\mathcal{T}}\left(\overline{\lambda}_{\mathcal{S}}+\overline{\lambda}_{\mathcal{S}\cup\mathcal{R}} \right)\bm{\overline{e}}_{\mathcal{S}} + \sum_{\mathcal{S}\subseteq\mathcal{T}}\overline{\lambda}_{\mathcal{S}\cup\mathcal{R}}(\bm{\overline{e}}_{\mathcal{S}\cup\mathcal{R}} -\bm{\overline{e}}_{\mathcal{S}}) \\
    &= \sum_{\mathcal{S}\subseteq\mathcal{T}}\left(\overline{\lambda}_{\mathcal{S}}+\overline{\lambda}_{\mathcal{S}\cup\mathcal{R}} \right)\bm{\overline{e}}_{\mathcal{S}} + \sum_{\mathcal{S}\subseteq\mathcal{T}}\overline{\lambda}_{\mathcal{S}\cup\mathcal{R}} \bm{\overline{e}}_{\mathcal{R}}. 
\end{align*}}%
The claim thus follows from the linear independence of $\bm{\overline{e}}_{\mathcal{S}}$, $\mathcal{S}\subseteq\mathcal{T}$, and $\bm{\overline{e}}_{\mathcal{R}}$, and the proof is complete.


\section{Proof of Proposition \ref{Prop:coreinJPO}}\label{App:coreinJPO}

Let $(c_1,\dots,c_n,c_R)\in core(\mathcal{N},\mathcal{B})$ and choose $\bm{\eta}_*=(\eta_{*,1},\dots,\eta_{*,n})$ such that $\omega_i=\omega_i(\bm{A}_*,\bm{p}_*,\bm{\eta}_*) = c_i$, for $i=1,\dots,n$. Note that
{\small\begin{align*}
0\leq c_{R}
= \mathcal{B}(\mathcal{N})-\sum_{i=1}^n\omega_i
&=\sum_{i=1}^n\left(\mu_i+\frac{\gamma_i\sigma_i^2}{2}-\rho_i(\bm{A}_*)\right)-\rho_R(\bm{p}_*)-\sum_{i=1}^n\left(\mu_i+\frac{\gamma_i\sigma_i^2}{2}-\rho_i(\bm{A}_*)-\pi_i(\eta_{*,i},p_{*,i})\right)\\
&=\sum_{i=1}^n\pi_i(\eta_{*,i},p_{*,i})-\rho_R(\bm{p}_*),
\end{align*}}%
which implies $\rho_R(\bm{p}_*)\leq \sum_{i=1}^n\pi_i(\eta_{*,i},p_{*,i})$, i.e., the reinsurer's IR \eqref{Eq:ReinsurerIR} is satisfied. In addition, for each $i=1,\dots,n$, it holds that $0\leq c_i=\omega_i=\mu_i+\frac{\gamma_i\sigma_i^2}{2}-\rho_i(\bm{A}_*)-\pi_i(\eta_{*,i},p_{*,i}),$
which implies members' IRs \eqref{Eq:MemberIR} are also satisfied. As $(\bm{A}_*,\bm{p}_*,\bm{\eta}_*)\in\mathcal{IR}$, its JP-optimality follows from Theorem \ref{Thm:JPO}.

\section{Proof of Proposition \ref{Prop:CoreiffCS}}\label{App:CoreiffCS}

Assume that $(\bm{A},\bm{p},\bm{\eta})$ is not coalitional stable.
This induces two cases. 

(i): There exists $\emptyset\neq\mathcal{C}\subsetneq\mathcal{N}$, $R\in\mathcal{C}$ and $(\check{\bm{A}}^\mathcal{C},\check{\bm{p}}^\mathcal{C},\check{\bm{\eta}}^\mathcal{C})\in\mathcal{F}^\mathcal{C}$ such that $u_i(\check{\bm{A}}^\mathcal{C},\check{\bm{p}}^\mathcal{C},\check{\bm{\eta}}^\mathcal{C}; \mathcal{C})\leq u_i(\bm{A},\bm{p},\bm{\eta}; \mathcal{N})$ for all $i\in\mathcal{C}$, and $v(\check{\bm{\eta}}^\mathcal{C},\check{\bm{p}}^\mathcal{C}; \mathcal{C})\leq v(\bm{\eta},\bm{p}; \mathcal{N})$,
with at least one strict inequality.
Note that $\mathcal{C}\neq\mathcal{N}$ since $(\bm{A},\bm{p},\bm{\eta})\in\mathcal{JPO}$. Hence,
{\small \begin{align*}
\sum_{i\in\mathcal{C}}\omega_i({\bm{A},\bm{p},\bm{\eta}})
=&\sum_{i\in\mathcal{C},i\neq R}\left(\mu_i+\frac{\gamma_i\sigma_i^2}{2}-{u_i(\bm{A},\bm{p},\bm{\eta}; \mathcal{N}) }\right)-v(\bm{\eta},\bm{p};{\mathcal{N}})\\
<&\sum_{i\in\mathcal{C},i\neq R}\left(\mu_i+\frac{\gamma_i\sigma_i^2}{2}{-u_i(\check{\bm{A}}^\mathcal{C},\check{\bm{p}}^\mathcal{C},\check{\bm{\eta}}^\mathcal{C}; \mathcal{C})}\right)-v(\check{\bm{\eta}}^\mathcal{C},\check{\bm{p}}^\mathcal{C}; \mathcal{C})\\
\leq&\sum_{i\in\mathcal{C},i\neq R}\left(\mu_i+\frac{\gamma_i\sigma_i^2}{2}\right)-\min_{(\bm{A}^\mathcal{C},\bm{p}^\mathcal{C},\bm{\eta}^\mathcal{C})\in \mathcal{F}^{\mathcal{C}}}{\sum_{i\in\mathcal{C}}} \left(\rho_i(\bm{A}^\mathcal{C}; \mathcal{C})+\rho_R(\bm{p}^\mathcal{C}; \mathcal{C}) \right)
=\mathcal{B}(\mathcal{C}). 
\end{align*}}%

(ii): There exists $\emptyset\neq\mathcal{C}\subsetneq\mathcal{N}$, $R\notin\mathcal{C}$ and $\check{\bm{A}}^\mathcal{C}\in\mathcal{F}^\mathcal{C}$ such that $ u_i(\check{\bm{A}}^\mathcal{C},\bm{0}^\mathcal{C},\bm{0}^\mathcal{C}; \mathcal{C})\leq u_i(\bm{A},\bm{p},\bm{\eta}; \mathcal{N})$ for all $i\in\mathcal{C}$,
with at least one strict inequality. Hence,
{\small \begin{align*}
\sum_{i\in\mathcal{C}}\omega_i({\bm{A},\bm{p},\bm{\eta}})
=\sum_{i\in\mathcal{C}}\left(\mu_i+\frac{\gamma_i\sigma_i^2}{2}-{u_i(\bm{A},\bm{p},\bm{\eta}; \mathcal{N}) }\right)
&<\sum_{i\in\mathcal{C}}\left(\mu_i+\frac{\gamma_i\sigma_i^2}{2}{-u_i(\check{\bm{A}}^\mathcal{C},\bm{0}^\mathcal{C},\bm{0}^\mathcal{C}; \mathcal{C})}\right)\\
& \leq\sum_{i\in\mathcal{C}}\left(\mu_i+\frac{\gamma_i\sigma_i^2}{2}\right)-\min_{\bm{A}^\mathcal{C}\in \mathcal{F}^{\mathcal{C}}}{\sum_{i\in\mathcal{C}}} \rho_i(\bm{A}^\mathcal{C}; \mathcal{C})
=\mathcal{B}(\mathcal{C}),
\end{align*}}%

As the above two cases indicate 
$(\omega_1({\bm{A},\bm{p},\bm{\eta}}),\dots,\omega_n({\bm{A},\bm{p},\bm{\eta}}),\mathcal{B}(\mathcal{N})-\sum_{i=1}^n\omega_i({\bm{A},\bm{p},\bm{\eta}})))\not\in core(\mathcal{N},\mathcal{B})$ and contradiction arises, the proof is complete.

\section{Proof of Proposition \ref{Prop:nonNegEta}}\label{App:nonNegEta}

 To prove the first statement, it suffices to show that $\omega_i(\bm{A}_*,\bm{p}_*,\max\{\bm{0},\bm{\eta}_{\min}(\bm{p}_*)\})\geq0$ for all $i\in \mathcal{N}$. 
Using \eqref{Eq:OptAP2}, the $i$-th row of $\bm{A}_*$ is given by 
   {\small  \begin{equation*}
        \bm{A}_{*,i} = \frac{\gamma_i^{-1}}{\sum_{j=1}^n\gamma_j^{-1}}\bm{1}^\top\bm{D}(\bm{1}-\bm{p}_*) + k(\alpha_{*,i} - \bar{\alpha}_{*,i})\bm{\mu}^\top \bm{\Sigma}^{-1}, 
    \end{equation*}}%
where $\alpha_{*,i}:= (1-p_{*,i})\mu_i$ and $\bar{\alpha}_{*,i} = \gamma_i^{-1}\sum_{j=1}^n (1-p_{*,j})\mu_j/\sum_{j=1}^n\gamma_j^{-1}$. Hence, 
   {\small  \begin{align}
    \label{eq:A:Sigma:A:pareto}
         \bm{A}_{*,i}\bm{\Sigma} \bm{A}_{*,i}^\top 
         &= \left(\frac{\gamma_i^{-1}}{\sum_{j=1}^n\gamma_j^{-1}}\right)^2 \bm{1}^\top \bm{D}(\bm{1}-\bm{p}_*) \bm{\Sigma} \bm{D}(\bm{1}-\bm{p}_*) \bm{1} +\frac{2k\gamma_i^{-1}(\alpha_{*,i}-\bar{\alpha}_{*,i})}{\sum_{j=1}^n\gamma_j^{-1}} \bm{\mu}^\top \bm{D}(\bm{1}-\bm{p}_*) \bm{1} \nonumber\\
         &+k^2(\alpha_{*,i}-\bar{\alpha}_{*,i})^2 \bm{\mu}^\top \bm{\Sigma}^{-1}\bm{\mu} \nonumber \\
         &= \left(\frac{\gamma_i^{-1}}{\sum_{j=1}^n\gamma_j^{-1}}\right)^2 \bm{1}^\top \bm{D}(\bm{1}-\bm{p}_*) \bm{\Sigma} \bm{D}(\bm{1}-\bm{p}_*) \bm{1}  + 2k\bar{\alpha}_{*,i}(\alpha_{*,i}-\bar{\alpha}_{*,i}) + k(\alpha_{*,i}-\bar{\alpha}_{*,i})^2\nonumber \\
         &= \left(\frac{\gamma_i^{-1}}{\sum_{j=1}^n\gamma_j^{-1}}\right)^2\sum_{j,l=1}^n(1-p_{*,j})(1-p_{*,l})\sigma_{jl} + k( (\alpha_{*,i})^2 - (\bar{\alpha}_{*,i})^2).
    \end{align}}%
Using this, we have, for $i=1,\dots,n$,
{\small\begin{align*}
        &\quad \omega_i(\bm{A}_*,\bm{p}_*,\bm{\eta}_{\min}(\bm{p}_*))\\ &= \frac{\gamma_i}{2}\left(\sigma_i^2 -    \bm{A}_{*,i}\bm{\Sigma} \bm{A}_{*,i}^\top\right)  -\eta_{\min,i}p_{*,i}\mu_i\\
        &= \frac{\gamma_i}{2}\left(\sigma_i^2 - \left(\frac{\gamma_i^{-1}}{\sum_{j=1}^n\gamma_j^{-1}}\right)^2\sum_{j,l=1}^n(1-p_{*,j})(1-p_{*,l})\sigma_{jl} - k( (\alpha_{*,i})^2 - (\bar{\alpha}_{*,i})^2)   \right)-\frac{\gamma_R}{2}p_{*,i}\sum_{m=1}^n\sigma_{im}p_{*,m} \\
        &\geq  \frac{\gamma_i}{2}\left(\sigma_i^2 - \left(\frac{\gamma_i^{-1}}{\sum_{j=1}^n\gamma_j^{-1}}\right)^2\sum_{j,l=1}^n (\sigma_{jl})_+ - k  (\alpha_{*,i})^2    \right) -\frac{\gamma_R}{2}\sum_{m=1}^n(\sigma_{im})_+\\
          &\geq  \frac{\gamma_i}{2}\left(\sigma_i^2 - \left(\frac{\gamma_i^{-1}}{\sum_{j=1}^n\gamma_j^{-1}}\right)^2\sum_{j,l=1}^n (\sigma_{jl})_+ - k  \mu_i^2   \right)-\frac{\gamma_R}{2}\sum_{m=1}^n(\sigma_{im})_+\geq0,
    \end{align*}}%
where the second equality results from \eqref{Eq:mineta},
the first inequality follows from the fact that $p_{*,i}\in[0,1]$ under \eqref{Eq:unicond2}, and the second inequality follows from \eqref{Eq:WGcond}. Following a similar argument, one can also obtain $\omega_i(\bm{A}_*,\bm{p}_*,\bm{0})\geq 0$. 

{To verify the reinsurer's welfare gain is non-negative if $\bm{\eta}\geq \max\{\bm{0},\bm{\eta}_{\min}(\bm{p}_*)\}$, it suffices to show that $\omega_R(\bm{A}_*,\bm{p}_*,\bm{\eta}_{\min}(\bm{p}_*))\geq 0$, since $\omega_R(\bm{A}_*,\bm{p}_*,\bm{\eta}_{\min}(\bm{p}_*)) \leq \omega_R(\bm{A}_*,\bm{p}_*,\bm{\eta})$. Indeed,} $\frac{\gamma_{R}}{2}\bm{p}_*^\top\bm{\Sigma}\bm{p}_*={ \bm{p}_*^\top\bm{D}(\bm{\mu})\bm{\eta}_{\min}(\bm{p}_*)}\leq(\bm{D}(\bm{\mu})\bm{\eta} )^\top\bm{p}_*.$ 
Therefore, the proof of the first statement is complete. 

For the second statement, note that for any $\bm{c}\in core(\mathcal{N},\mathcal{B})$, we have $\mathcal{B}(\mathcal{N})-c_i=\sum_{j\neq i}c_i\geq \mathcal{B}(\mathcal{N}\backslash\{i\})$ for all $i=1,\dots,n$.
This leads to $c_i\leq \mathcal{B}(\mathcal{N})-\mathcal{B}(\mathcal{N}\backslash\{i\})$.  Then, imposing \eqref{Eq:coreBound}
means each member has an excess welfare gain to pay $\max\{\bm{0},\bm{\eta}_{\min}(\bm{p}_*)\}$ after joining the program.
Then, by invoking Proposition \ref{Prop:coreinJPO}, we can construct a JP-optimal contract with nonnegative loading. 


\section{Proof of Corollary \ref{Cor:JPOSingleLoading}}\label{App:JPOSingleLoading}
Let $\bm{\eta}^{(1)}$ and $\bm{\eta}^{(2)}$ be the safety loadings given in the statement. Define $\underline{n}:=\max_i\eta_i^{(2)}$ and $\overline{n}:=\min_i\eta_i^{(1)}$. For any scalar $t\in[\underline{n},\overline{n}]$, set a uniform loading $t\bm{1}$, and define the associated welfare gains by $\bm{\omega}(t)$, i.e., $\omega_i(t):=\omega_i(\bm{A}_*,\bm{p}_*,t\bm{1})=\omega_i(\bm{A}_*,\bm{p}_*,\bm{0})-tp_{*,i}\mu_i$ for all $i=1,\dots,n$, and $\omega_R(t):=\omega_R(\bm{A}_*,\bm{p}_*,t\bm{1}) = \mathcal{B}(\mathcal{N})-\sum_{i=1}^n\omega_i(t)$. 
Note that the efficiency constraint is satisfied as $\sum_{i\in\mathcal{N}}\omega_i(t)=\mathcal{B}(\mathcal{N})$.

Next, we show that $\omega_i(t)\geq0$ for all $i\in\mathcal{N}$. As $t\leq\overline{n}\leq\eta_i^{(1)}$ for all $i=1,\ldots,n$, we have
\begin{equation*}
    \omega_i(t)=\omega_i(0)-tp_{*,i}\mu_i\geq\omega_i(0)-\eta_i^{(1)}p_{*,i}\mu_i=c_i^{(1)}:= \omega_i(\bm{A}_*,\bm{p}_*,\bm{\eta}^{(1)}).
\end{equation*}
Similarly, since $t\leq\underline{n}\leq\eta_i^{(2)}$ for all $i=1,\ldots,n$, we obtain $\omega_i(t)\leq\omega_i(0)-\eta_i^{(2)}p_{*,i}\mu_i=c_i^{(2)} := \omega_i(\bm{A}_*,\bm{p}_*,\bm{\eta}^{(2)})$. 
These two inequalities lead to $ \omega_i(t)\geq c_i^{(1)}\geq0$ and  $\omega_R(t)=\mathcal{B}(\mathcal{N})-\sum_{i=1}^n\omega_i(t)\geq \mathcal{B}(\mathcal{N})-\sum_{i=1}^nc_i^{(2)}=\omega_R^{(2)}\geq0.$
Therefore, $\bm{\omega}(t)\in\mathbb{R}_+^{n+1}$. 

The last step is to verify all inequality constraints in the definition of $core(\mathcal{N},\mathcal{B})$ (see \eqref{Eq:CoreNB}). For any $\mathcal{C}\subsetneq\mathcal{N}$ such that $R\notin\mathcal{C}$, we have $\sum_{i\in\mathcal{C}}\omega_i(t)\geq\sum_{i\in\mathcal{C}}c_i^{(1)}\geq\mathcal{B}(\mathcal{C})$,
since $\bm{\omega}^{(1)}\in core(\mathcal{N},\mathcal{B})$. Similarly, for any $\mathcal{C}\subsetneq\mathcal{N}$ such that $R\in\mathcal{C}$, we have $\sum_{i\in\mathcal{C}}\omega_i(t)=\mathcal{B}(\mathcal{N})-\sum_{i\notin\mathcal{C}}\omega_i(t)\geq \mathcal{B}(\mathcal{N})-\sum_{i\notin\mathcal{C}}c_i^{(2)}=\sum_{i\in\mathcal{C}}c_i^{(2)}\geq\mathcal{B}(\mathcal{C}),$
since $\bm{c^{(2)}}\in core(\mathcal{N},\mathcal{B})$.
Therefore, $\bm{\omega}(t)\in core(\mathcal{N},\mathcal{B})$.

\section{Proof of Proposition \ref{Prop:leader}}\label{App:leader}
Using $\bm{p}^*$ in \eqref{Eq:OptAP}, we have 
{\small\begin{align*}
\bm{\eta}^\top \bm{D}(\bm{\mu}) \bm{p}^*
= \bm{\mu}^\top\bm{\eta}-\bm{\eta}^\top\bm{D}(\bm{\mu})\bm{M}^{-1}\bm{D}(\bm{\mu})\bm{\eta}, \quad 
(\bm{p}^*)^\top\bm{\Sigma}\bm{p}^*
= (\bm{1}-\bm{M}^{-1}\bm{D}(\bm{\mu})\bm{\eta})^\top\bm{\Sigma}(\bm{1}-\bm{M}^{-1}\bm{D}(\bm{\mu})\bm{\eta}).
\end{align*}}%
As the terms independent of $\bm{\eta}$ play no role in the optimization, Problem \eqref{Prob:leader2} is equivalent to minimizing 
{\small \begin{align*}
&\frac{\gamma_{R}}{2}(-2\cdot\bm{1}^\top\bm{\Sigma}\bm{M}^{-1}\bm{D}(\bm{\mu})\bm{\eta}+\bm{\eta}^\top\bm{D}(\bm{\mu})\bm{M}^{-1}\bm{\Sigma}\bm{M}^{-1}\bm{D}(\bm{\mu})\bm{\eta})-\bm{\mu}^\top\bm{\eta}+\bm{\eta}^\top\bm{D}(\bm{\mu})\bm{M}^{-1}\bm{D}(\bm{\mu})\bm{\eta}\\
=&\ \bm{\eta}^\top\left(\bm{D}(\bm{\mu})\bm{M}^{-1}\bm{D}(\bm{\mu})+\frac{\gamma_{R}}{2}\bm{D}(\bm{\mu})\bm{M}^{-1}\bm{\Sigma}\bm{M}^{-1}\bm{D}(\bm{\mu})\right)\bm{\eta}-(\bm{\mu}^\top+\gamma_{R}\bm{1}^\top\bm{\Sigma}\bm{M}^{-1}\bm{D}(\bm{\mu}))\bm{\eta}
\end{align*}}
Then, the FOC yields 
{\small\begin{align*}
    &\bm{D}(\bm{\mu})(\gamma_{R}\bm{M}^{-1}\bm{\Sigma}+2\bm{I})\bm{M}^{-1}\bm{D}(\bm{\mu})\bm{\eta}=\gamma_{R}\bm{D}(\bm{\mu})\bm{M}^{-1}\bm{\Sigma}\bm{1}+\bm{\mu}    
    \iff {(\gamma_{R}\bm{\Sigma}+2\bm{M})\bm{M}^{-1}\bm{D}(\bm{\mu})\bm{\eta}= (\gamma_{R}\bm{\Sigma}+\bm{M})\bm{1}}.
\end{align*}}%
Since $\bm{\Sigma}$ and $\bm{M}$ are positive definite (see Lemma \ref{lem:M}), $\gamma_{R}\bm{\Sigma}+2\bm{M}$ is invertible. The result thus follows.


Finally, we prove the non-negativity of $\bm{\eta}^*$ under Condition \eqref{Eq:deltaINE}. Consider 
{\small\begin{align}
        \label{eq:Dmu:eta:0}
\bm{D}(\bm{\mu}) \bm{\eta}^* &=  \bm{M}\bm{1}  -  \bm{M}(\gamma_{R}\bm{\Sigma}+2\bm{M})^{-1}\bm{M} \bm{1} \nonumber \\
            &= \bm{M}\bm{1} -\frac{1}{2}\left(\gamma_R\bm{\Sigma} + 2\bm{M}-\gamma_R\bm{\Sigma} \right) (\gamma_{R}\bm{\Sigma}+2\bm{M})^{-1}\bm{M}\bm{1} \nonumber \\
            &= \frac{1}{2}\bm{M}\bm{1}  + \frac{\gamma_R}{2}\bm{\Sigma}(\gamma_R\bm{\Sigma}+2\bm{M})^{-1} \bm{M}\bm{1} \nonumber \\
            &= \frac{1}{2}\bm{M}\bm{1}  + \frac{\gamma_R}{4}\bm{\Sigma}(\gamma_R\bm{\Sigma}+2\bm{M})^{-1}\left(\gamma_R\bm{\Sigma} + 2\bm{M} - \gamma_R\bm{\Sigma} \right)\bm{1}  \nonumber \\
            &= \frac{1}{2}\bm{M}\bm{1}  + \frac{\gamma_R}{4}\bm{\Sigma}\bm{1} - \frac{\gamma_R^2}{4}\bm{\Sigma}(\gamma_R\bm{\Sigma}+2\bm{M})^{-1}\bm{\Sigma}\bm{1} 
            :=\frac{1}{2}\bm{B}\bm{1}.
\end{align}}%
Thus, we only need to show $\bm{B}\bm{1}\geq\bm{0}$, implied by $B_{ii}\geq\sum_{j\neq i}|B_{ij}|$ for all $i=1,\ldots,n$, where $B_{ij}$ represents the $(i,j)$-th entry of $\bm{B}$. 

Define $\bm{B}^0:=\bm{M}+\frac{\gamma_R}{2}\bm{\Sigma}$ and $\bm{B}^1:=\bm{\Sigma}(\gamma_R\bm{\Sigma}+2\bm{M})^{-1}\bm{\Sigma}$, then
{\small\begin{align*}
B_{ii}-\sum_{j\neq i}|B_{ij}|
=&B_{ii}^0-\frac{\gamma_R^2}{2}B_{ii}^1-\sum_{j\neq i}\left|B_{ij}^0-\frac{\gamma_R^2}{2}B_{ij}^1\right|
\geq  B_{ii}^0-\sum_{j\neq i}|B_{ij}^0|-\frac{\gamma_R^2}{2}\left(B_{ii}^1+\sum_{j\neq i}|B_{ij}^1|\right),
\end{align*}}%
{which is non-negative if}
\begin{equation}\label{Eq:INEB0B1}
    \min_{i\in\{1,\ldots,n\}}\left(B_{ii}^0-\sum_{j\neq i}|B_{ij}^0|\right)
    \geq\frac{\gamma_R^2}{2}\|\bm{B}^1\|_{\infty}.
\end{equation}

We complete the proof by showing that \eqref{Eq:deltaINE} leads to \eqref{Eq:INEB0B1}. For the left-hand side of \eqref{Eq:INEB0B1}, we have
{\small\begin{equation*}
\min_{i\in\{1,\ldots,n\}}\left(B_{ii}^0-\sum_{j\neq i}|B_{ij}^0|\right)
=\min_{i\in\{1,\ldots,n\}}\left(M_{ii}+\frac{\gamma_R}{2}\sigma_{ii}-\sum_{j\neq i}\left|M_{ij}+\frac{\gamma_R}{2}\sigma_{ij}\right|\right)\geq\delta_{\bm{M}}+\frac{\gamma_R}{2}\delta_{\bm{\Sigma}}. 
\end{equation*}}%
For the right-hand side of \eqref{Eq:INEB0B1}, we observe that {\small  
\begin{align*}
\|\bm{B}^1\|_{\infty}
=\|\bm{\Sigma}(\gamma_R\bm{\Sigma}+2\bm{M})^{-1}\bm{\Sigma}\|_{\infty}
\leq\|\bm{\Sigma}\|_{\infty}^2\|(\gamma_R\bm{\Sigma}+2\bm{M})^{-1}\|_{\infty}
\leq\frac{\|\bm{\Sigma}\|_{\infty}^2}{\kappa_{\min}}
\leq\frac{\|\bm{\Sigma}\|_{\infty}^2}{2\delta_{\bm{M}}+\gamma_R\delta_{\bm{\Sigma}}},
\end{align*}}%
where $\kappa_{\min}:=\min_{i\in\{1,\ldots,n\}}(\gamma_R\sigma_{ii}+2M_{ii}-\sum_{j\neq i}|\gamma_R\sigma_{ij}+2M_{ij}|)$ and the last inequality follows from the triangle inequality. The proof is complete.

\section{Proof of Lemma \ref{Lem:reinsurerIR}}\label{App:reinsurerIR}

Using the expression $\bm{p}^*(\bm{\eta}^*) = \bm{1}-\bm{M}^{-1}\bm{D}(\bm{\mu})\bm{\eta}^*$, we have $\omega_R^*:= \omega_R(\bm{A}^*(\bm{\eta}^*),\bm{p}^*(\bm{\eta}^*),\bm{\eta}^*)$ satisfies 
{\small \begin{align*}
  \omega_R^*  & =  (\bm{\eta}^*)^\top \bm{D}(\bm{\mu})\bm{p}^* - \frac{\gamma_R}{2}(\bm{p}^*)^\top\bm{\Sigma}\bm{p}^*\\ 
    &=  (\bm{\eta}^*)^\top \bm{D}(\bm{\mu})\left[\bm{1}-\bm{M}^{-1}\bm{D}(\bm{\mu})\bm{\eta}^*\right] - \frac{\gamma_R}{2}[\bm{1}-\bm{M}^{-1}\bm{D}(\bm{\mu})\bm{\eta}^*]^\top\bm{\Sigma}[\bm{1}-\bm{M}^{-1}\bm{D}(\bm{\mu})\bm{\eta}^*] \\
    &= - (\bm{D}(\bm{\mu})\bm{\eta}^*)^\top\bm{M}^{-1}\bm{D}(\bm{\mu})\bm{\eta}^* - \frac{\gamma_R}{2}(\bm{D}(\bm{\mu})\bm{\eta}^*)^\top\bm{M}^{-1}\bm{\Sigma}\bm{M}^{-1}\bm{D}(\bm{\mu})\bm{\eta}^* \\
    &\quad + \bm{1}^\top\bm{D}(\bm{\mu})\bm{\eta}^*+\gamma_R\bm{1}^\top\bm{\Sigma}\bm{M}^{-1}\bm{D}(\bm{\mu})\bm{\eta}^*  - \frac{\gamma_R}{2} \bm{1}^\top\bm{\Sigma}\bm{1} \\
    &= -\frac{1}{2}(\bm{D}(\bm{\mu})\bm{\eta}^*)^\top \left[  \gamma_R\bm{M}^{-1}\bm{\Sigma}  + 2\bm{I}_n  \right]\bm{M}^{-1} \bm{D}(\bm{\mu})\bm{\eta}^* + \left[\left(\gamma_R\bm{M}^{-1}\bm{\Sigma} +\bm{I}_n   \right)\bm{1}\right]^\top \bm{D}(\bm{\mu})\bm{\eta}^* - \frac{\gamma_R}{2} \bm{1}^\top\bm{\Sigma}\bm{1}.
\end{align*}}%
Using \eqref{Eq:Opteta}, we have $\bm{D}(\bm{\mu})\bm{\eta}^*= \bm{M}(\gamma_{R}\bm{M}^{-1}\bm{\Sigma}+2\bm{I}_n)^{-1}(\gamma_{R}\bm{M}^{-1}\bm{\Sigma}+\bm{I}_n)\bm{1}$, whence
  {\small   \begin{align*}
       \omega_R^*
         &= -\frac{1}{2}\bm{1}^\top(\gamma_{R}\bm{M}^{-1}\bm{\Sigma}+\bm{I}_n)^\top[(\gamma_{R}\bm{M}^{-1}\bm{\Sigma}+2\bm{I}_n)^{-1}]^\top\bm{M}(\gamma_{R}\bm{M}^{-1}\bm{\Sigma}+\bm{I}_n)\bm{1} \\
         &\quad + \bm{1}^\top(\gamma_{R}\bm{M}^{-1}\bm{\Sigma}+\bm{I}_n)^\top\bm{M}(\gamma_{R}\bm{M}^{-1}\bm{\Sigma}+2\bm{I}_n)^{-1}(\gamma_{R}\bm{M}^{-1}\bm{\Sigma}+\bm{I}_n)\bm{1} - \frac{\gamma_R}{2} \bm{1}^\top\bm{\Sigma}\bm{1} \\
         &= \frac{1}{2}\bm{1}^\top \left[(\gamma_{R}\bm{M}^{-1}\bm{\Sigma}+\bm{I}_n)^\top\bm{M}(\gamma_{R}\bm{M}^{-1}\bm{\Sigma}+2\bm{I}_n)^{-1}(\gamma_{R}\bm{M}^{-1}\bm{\Sigma}+\bm{I}_n) - \gamma_R \bm{\Sigma} \right]\bm{1}.
    \end{align*}}%
Using $\gamma_R\bm{\Sigma} =  -\bm{M}+(\gamma_{R}\bm{M}^{-1}\bm{\Sigma}+\bm{I}_n)^\top\bm{M}(\gamma_{R}\bm{M}^{-1}\bm{\Sigma}+\bm{I}_n)^{-1}(\gamma_{R}\bm{M}^{-1}\bm{\Sigma}+\bm{I}_n),$
we have 
   {\small  \begin{align*}
    & \ \ \ \ (\gamma_{R}\bm{M}^{-1}\bm{\Sigma}+\bm{I}_n)^\top\bm{M}(\gamma_{R}\bm{M}^{-1}\bm{\Sigma}+2\bm{I}_n)^{-1}(\gamma_{R}\bm{M}^{-1}\bm{\Sigma}+\bm{I}_n) - \gamma_R \bm{\Sigma}\\
    &= \bm{M} + (\gamma_{R}\bm{M}^{-1}\bm{\Sigma}+\bm{I}_n)^\top\bm{M}\left[(\gamma_{R}\bm{M}^{-1}\bm{\Sigma}+2\bm{I}_n)^{-1}-(\gamma_{R}\bm{M}^{-1}\bm{\Sigma}+\bm{I}_n)^{-1}\right](\gamma_{R}\bm{M}^{-1} \bm{\Sigma}+\bm{I}_n)\\
    &= \bm{M} - (\gamma_{R}\bm{M}^{-1}\bm{\Sigma}+\bm{I}_n)^\top\bm{M}(\gamma_{R}\bm{M}^{-1}\bm{\Sigma}+2\bm{I}_n)^{-1}  \\
    &= \left[ \bm{M}(\gamma_{R}\bm{M}^{-1}\bm{\Sigma}+2\bm{I}_n) - (\gamma_{R}\bm{M}^{-1}\bm{\Sigma}+\bm{I}_n)^\top\bm{M} \right](\gamma_{R}\bm{M}^{-1}\bm{\Sigma}+2\bm{I}_n)^{-1} \\
    &= \left(\gamma_{R}\bm{M}^{-1}\bm{\Sigma}\bm{M}^{-1}+2\bm{M}^{-1} \right)^{-1}.
    \end{align*}}%
Therefore, $\omega_R^*=\frac{1}{2}\bm{1}^\top \left(\gamma_{R}\bm{M}^{-1}\bm{\Sigma}\bm{M}^{-1}+2\bm{M}^{-1} \right)^{-1} \bm{1}$. 
 Since $\gamma_R\bm{M}^{-1}\bm{\Sigma}\bm{M}^{-1} + 2\bm{M}^{-1}$ is positive definite, we conclude that   $\omega_R^*=\omega_R(\bm{A}^*(\bm{\eta}^*),\bm{p}^*(\bm{\eta}^*),\bm{\eta}^*)>0$.  

Next, we verify the members' IR constraints given \eqref{Eq:unicond} and \eqref{Eq:MIRcond}. Using \eqref{Eq:OptAP}, the fact that under \eqref{Eq:unicond}, $p^*_{i}\in[0,1]$, $i=1,\dots,n$, and following the derivation of \eqref{eq:A:Sigma:A:pareto}, $\omega_i^*:= \omega_i(\bm{A}^*(\bm{\eta}^*),\bm{p}^*(\bm{\eta}^*),\bm{\eta}^*)$ satisfies  
   {\small  \begin{align*}
        \omega_i^*
        &= \frac{\gamma_i}{2}\left(\sigma_i^2 - \left(\frac{\gamma_i^{-1}}{\sum_{j=1}^n\gamma_j^{-1}}\right)^2\sum_{j,l=1}^n(1-p^*_{j})(1-p^*_{l})\sigma_{jl} - k( (\alpha^*_{i})^2 - (\bar{\alpha}^*_{i})^2)   \right) - \eta^*_ip^*_i\mu_i \\
          &\geq  \frac{\gamma_i}{2}\left(\sigma_i^2 - \left(\frac{\gamma_i^{-1}}{\sum_{j=1}^n\gamma_j^{-1}}\right)^2\sum_{j,l=1}^n (\sigma_{jl})_+ - k  \mu_i^2   \right)- \eta^*_ip^*_i\mu_i,
    \end{align*}}%
where $\alpha^*_i:= (1-p^*_i)\mu_i$ and $\bar{\alpha}^*_i = \gamma_i^{-1}\sum_{j=1}^n (1-p^*_j)\mu_j/\sum_{j=1}^n\gamma_j^{-1}$.
 Using  the upper bound of $\mu_i\eta^*_i$ in  \eqref{Eq:unicond},  $p^*_i \in[0,1]$, and  \eqref{Eq:MIRcond},  we further have  
 {\small  \begin{align*}
          \omega_i^* &\geq   \frac{\gamma_i}{2}\left(\sigma_i^2 - \left(\frac{\gamma_i^{-1}}{\sum_{j=1}^n\gamma_j^{-1}}\right)^2\sum_{j,l=1}^n (\sigma_{jl})_+ - k  \mu_i^2   \right) - k\mu_i^2\gamma_i  + \frac{\sum_{m\neq i} (k\mu_i\mu_m-\sigma_{im})_+   + k\mu_i^2 - \sigma_i^2 }{\sum_{j=1}^n\gamma_j^{-1}} \geq 0. 
    \end{align*}}%
Therefore, the members' IR constraints are fulfilled, and the proof is complete.

\section{Proof of Corollary \ref{Lem:singleprice}}\label{App:singleprice}

When $\bm{\eta} = \eta\bm{1}$, the optimal reinsurance strategy in Proposition \ref{Prop:convexfollower} becomes
$\bm{p} = \bm{1} - \eta \bm{M}^{-1}\bm{\mu}.$
We shall repeat the proof of Proposition \ref{Prop:leader} with this new setting.

By direct substitutions, we have $\bm{p}^\top\bm{\Sigma}\bm{p}=
\bm{1}^\top\bm{\Sigma}\bm{1}-2\eta\bm{1}^\top\bm{\Sigma}\bm{M}^{-1}\bm{\mu}+\eta^2\bm{\mu}^\top\bm{M}^{-1}\bm{\Sigma}\bm{M}^{-1}\bm{\mu}$ and $\eta\bm{\mu}^\top\bm{p}=\eta\bm{\mu}^\top\bm{1}-\eta^2\bm{\mu}^\top\bm{M}^{-1}\bm{\mu}$.
Removing terms that are independent of $\eta$, we see that Problem \eqref{Prob:leader2} is equivalent to minimizing 
{\small\begin{align*}
&\frac{\gamma_{R}}{2}(-2\eta\bm{1}^\top\bm{\Sigma}\bm{M}^{-1}\bm{\mu}+\eta^2\bm{\mu}^\top\bm{M}^{-1}\bm{\Sigma}\bm{M}^{-1}\bm{\mu})-\eta\bm{\mu}^\top\bm{1}+\eta^2\bm{\mu}^\top\bm{M}^{-1}\bm{\mu}\\
=&\left(\bm{\mu}^\top\bm{M}^{-1}\bm{\mu}+\frac{\gamma_{R}}{2}\bm{\mu}^\top\bm{M}^{-1}\bm{\Sigma}\bm{M}^{-1}\bm{\mu}\right)\eta^2-(\bm{\mu}^\top\bm{1}+\gamma_{R}\bm{1}^\top\bm{\Sigma}\bm{M}^{-1}\bm{\mu})\eta
\end{align*}}%
The FOC then yields the solution \eqref{Eq:Opteta} by noting the positive definiteness of $\bm{M}$ and $\bm{\Sigma}$. The second statement follows from Theorem \ref{Thm:BO}.

For the last statement, note that the non-negativity of $\eta^*$ follows from the diagonal dominance of $\bm{I}_n+\gamma_{R}\bm{M}^{-1}\bm{\Sigma}$ and its diagonal elements are positive. To prove the latter, we need $\gamma_R\|\bm{M}^{-1}\bm{\Sigma}\|_{\infty}\leq 1$. As $\bm{M}$ is strictly diagonally dominant, using the Varah bound, we have $\|\bm{M}^{-1}\|_{\infty}\leq 1/\delta_{\bm{M}}$. Together with $\gamma_R\leq\delta_{\bm{M}}/\|\bm{\Sigma}\|_{\infty}$, the result follows.

\section{A Discussion on Condition \eqref{Eq:unicond}}
\label{sec:discussion:eta:bound}

This note provides a discussion on obtaining bounds on the optimal safety loading $\bm{\eta}^*$, and consequently, a way to verify the condition \eqref{Eq:unicond} under the Bowley optimum by deriving an explicit representation that depends on the underlying model parameters. 

We first consider the benchmark case $\gamma_R=0$ where $\bm{\eta}^*$ admits a closed-form, entry-wise expression, allowing \eqref{Eq:unicond} to be written explicitly in terms of the model parameters. When $\gamma_R>0$, $\bm{\eta}^*$ contains additional terms involving matrix inverses, and its entries are no longer available in closed form. We therefore discuss ways to bound these terms and substitute the resulting bounds back into \eqref{Eq:unicond} to obtain tractable sufficient conditions. To this end, we utilize the formula derived in \eqref{eq:Dmu:eta:0}:
   {\small  \begin{equation}
                \label{eq:Dmu:eta:1}
       \bm{D}(\bm{\mu})\bm{\eta}^* =   \frac{1}{2}\bm{M}\bm{1}  + \frac{\gamma_R}{4}\bm{\Sigma}\bm{1} - \frac{\gamma_R^2}{4}\bm{\Sigma}(\gamma_R\bm{\Sigma}+2\bm{M})^{-1}\bm{\Sigma}\bm{1} .
    \end{equation}}%

   When the reinsurer is risk-neutral, i.e.,  $\gamma_R=0$, we have $  \bm{D}(\bm{\mu}) \bm{\eta}^* = \frac{1}{2}\bm{M}\bm{1}$.     Hence, 
      {\small   \begin{equation}
        \label{eq:mui:etai:gam_R=0}
         \mu_i \eta^*_i =    \frac{1}{2}\left[\frac{\sum_{j=1}^n(\sigma_{ij}-k\mu_i\mu_j) }{\sum_{j=1}^n\gamma_j^{-1}} + k\mu_i^2\gamma_i   \right]. 
        \end{equation}}%
    Substituting this into \eqref{Eq:unicond}, the condition is equivalent to, for any $i=1,\dots,n$,
     {\small 
        \begin{equation*} 
           k\mu_i^2\gamma_i + \frac{\sigma_i^2-k\mu_i^2}{\sum_{j=1}^n\gamma_j^{-1}}> \frac{\sum_{j\neq i}|\sigma_{ij}-k\mu_i\mu_j|}{\sum_{j=1}^n\gamma_j^{-1}} \quad \text{and} \quad    k\mu_i^2\gamma_i + \frac{\sigma_i^2-k\mu_i^2}{\sum_{j=1}^n\gamma_j^{-1}} >    \frac{\sum_{j\neq i}(\sigma_{ij}-k\mu_i\mu_j)_+ }{\sum_{j\neq i}\gamma_j^{-1}}   .
        \end{equation*}}%
 Combining the two inequalities, we see that  \eqref{Eq:unicond} is equivalent to \eqref{Eq:unicond2:gam_R=0}.


    When $\gamma_R>0$, 
    one has to handle the term $(\gamma_R\bm{\Sigma}+2\bm{M})^{-1}$. In that case, using \eqref{eq:Dmu:eta:1}, we have 
    {\small\begin{align}
    \label{eq:mui:etai:gam_R>0}
         \mu_i \eta^*_i &=     \frac{1}{2}\left[\frac{\sum_{j=1}^n(\sigma_{ij}-k\mu_i\mu_j) }{\sum_{j=1}^n\gamma_j^{-1}} + k\mu_i^2\gamma_i   \right] + \frac{\gamma_R}{4}\sum_{j=1}^n\sigma_{ij} - \frac{\gamma_R^2}{4} \left[\bm{\Sigma}(\gamma_R\bm{\Sigma}+2\bm{M})^{-1} \bm{\Sigma}\bm{1}   \right]_i,
    \end{align}}%
    where $[\bm{x}]_i$ denotes the $i$-th entry of $\bm{x}$.  The formula shows that the non-explicit term appears only at second order in $\gamma_R$. Substituting this into \eqref{Eq:unicond}, we see that the condition is equivalent to
   {\small\begin{equation}\label{Eq:unicond2:gam_R>0}
  \begin{aligned}
               k\mu_i^2\gamma_i + \frac{\sigma_i^2-k\mu_i^2}{\sum_{j=1}^n\gamma_j^{-1}} &>  \max\bigg\{ \frac{\sum_{j\neq i}|\sigma_{ij}-k\mu_i\mu_j|}{\sum_{j=1}^n\gamma_j^{-1}} - \frac{\gamma_R}{2}\sum_{j=1}^n\sigma_{ij} +  \frac{\gamma_R^2}{2} \left[\bm{\Sigma}(\gamma_R\bm{\Sigma}+2\bm{M})^{-1} \bm{\Sigma}\bm{1}   \right]_i,  \\
               &\qquad \frac{\sum_{j=1}^n(\sigma_{ij}-k\mu_i\mu_j) }{\sum_{j=1}^n\gamma_j^{-1}} + \frac{\gamma_R}{2}\sum_{j=1}^n\sigma_{ij} - \frac{\gamma_R^2}{2} \left[\bm{\Sigma}(\gamma_R\bm{\Sigma}+2\bm{M})^{-1} \bm{\Sigma}\bm{1}   \right]_i  \bigg\}.  
    \end{aligned}
\end{equation}}%

Herein, we provide two possible entry-wise bounds. For the first bound, we impose  the following condition: for $i=1,\dots,n$,
   {\small \begin{equation}
    \label{eq:delta:i}
       \kappa_i :=  \left(\gamma_R + \frac{2}{\sum_{j=1}^n\gamma_j^{-1}} \right)\left(\sigma_i^2 -\sum_{j\neq i}|\sigma_{ij}| \right) + 2k\mu_i^2\left(\gamma_i-\frac{1}{\sum_{j=1}^n\gamma_j^{-1}}\right)   +\frac{2k\mu_i}{\sum_{j=1}^n\gamma_j^{-1}}\sum_{j\neq i}\mu_j>0,
    \end{equation}}%
i.e., the matrix $\gamma_R\bm{\Sigma}+2\bm{M}$ is strictly diagonally dominant. Let also $\kappa_{\min} := \min_{i=1,\dots,n}\kappa_i$. By the Varah bound for strictly diagonally dominant matrices, $ \|(\gamma_R\bm{\Sigma}+2\bm{M})^{-1}\|_\infty \leq \frac{1}{\kappa_{\min}}$.
Therefore, 
   {\small  \begin{align}
    \label{eq:gam_R:pertub:1}
      & \ \ \ \   \left|\gamma_R^2  \left[ \bm{\Sigma}(\gamma_R\bm{\Sigma}+2\bm{M})^{-1} \bm{\Sigma}\bm{1}   \right]_i\right| \leq  \frac{\gamma_R^2 \|\bm{\Sigma}\|_\infty }{\kappa_{\min}} \max_{j=1,\dots,n} \left|  \sum_{k=1}^n\sigma_{jk} \right|. 
    \end{align}}%
We can thus substitute \eqref{eq:gam_R:pertub:1} into \eqref{Eq:unicond2:gam_R>0} to obtain an expression for the condition \eqref{Eq:unicond}. 

In the second bound, by considering $ \bm{\Sigma} (\gamma_R\bm{\Sigma}+2\bm{M})^{-1} =   \bm{\Sigma}^{\frac{1}{2}} (\gamma_R\bm{I}_n+2\bm{\Sigma}^{-\frac{1}{2}}\bm{M}\bm{\Sigma}^{-\frac{1}{2}})^{-1}\bm{\Sigma}^{-\frac{1}{2}}$,
we have 
    \begin{equation*}
        \|\bm{\Sigma} (\gamma_R\bm{\Sigma}+2\bm{M})^{-1}\|_2 \leq \frac{1}{\gamma_R + 2\lambda_{\min}(\bm{\Sigma}^{-\frac{1}{2}}\bm{M}\bm{\Sigma}^{-\frac{1}{2}})^{-1}\bm{\Sigma}^{-\frac{1}{2}} ) }.
    \end{equation*}
Hence,
\begin{align}
        \label{eq:gam_R:pertub:2}
        & \ \ \ \ \left| \gamma_R^2 \left[\bm{\Sigma}(\gamma_R\bm{\Sigma}+2\bm{M})^{-1} \bm{\Sigma}\bm{1}   \right]_i\right| \leq \frac{\gamma_R^2}{\gamma_R + 2\lambda_{\min}(\bm{\Sigma}^{-\frac{1}{2}}\bm{M}\bm{\Sigma}^{-\frac{1}{2}})^{-1}\bm{\Sigma}^{-\frac{1}{2}} ) } \left\|  \bm{\Sigma}\bm{1} \right\|_2.
    \end{align}

Compared with \eqref{eq:gam_R:pertub:1}, the bound \eqref{eq:gam_R:pertub:2} does not require additional conditions as introduced in \eqref{eq:delta:i}. However, it could suffer from the curse of dimensionality as the $l_2$-norm of $\bm{\Sigma}\bm{1}$ would, in general, introduce a factor of $\sqrt{n}$.

\end{document}